\documentclass[letterpaper, 10pt, conference]{ieeeconf} 
\usepackage[utf8]{inputenc}
\usepackage[T1]{fontenc}
\usepackage{hyperref}
\usepackage{amssymb}
\usepackage{amsmath}
\usepackage[capitalize]{cleveref}
\usepackage{cancel}
\usepackage{color}
\usepackage{booktabs}
\usepackage{longtable}
\usepackage{xcolor}
\usepackage{scrextend}
\usepackage[]{mathtools}
\newcommand{\abs}[1]{\left| #1 \right|}
\usepackage{arydshln}

\usepackage{bbm}
\newtheorem{theorem}{Theorem}
\newtheorem{lemma}{Lemma}

\usepackage{capt-of}
\usepackage{graphicx}
\usepackage{booktabs}
\newcommand{\lb}{\left(}
\newcommand{\rb}{\right)}

\newcommand{\romn}{\rule{0em}{0.9em}}

\newcommand{\req}[1]{(\ref{#1.eq})} 
\newcommand{\be}{\begin{equation}}
\newcommand{\ee}{\end{equation}}

\newcommand{\R}{\mathbb{R}}
\newcommand{\sT}{{\scriptstyle T}}

\newcommand{\lcb}{\left\{}
\newcommand{\rcb}{\right\}}
\newcommand{\bone}{\mathbbm{1}}
\newcommand{\sOmega}{{\scriptstyle \Omega}}

\newcommand{\rmn}{{\rm n}}
\newcommand{\hstm}{\hspace{2em}}

\title{Parametric Resonance in Networked Oscillators}
\author{Karthik Chikmagalur and Bassam Bamieh}
\date{}
\providecolor{url}{HTML}{0077bb}
\providecolor{link}{HTML}{882255}
\providecolor{cite}{HTML}{999933}
\hypersetup{
  pdfauthor={Karthik Chikmagalur and Bassam Bamieh},
  pdftitle={Parametric Resonance in Networked Oscillators},
  pdfkeywords={},
  pdfsubject={},
  pdfcreator={Emacs 29.3 (Org mode 9.7-pre)},
  pdflang={English},
  breaklinks=true,
  colorlinks=true,
  linkcolor=link,
  urlcolor=url,
  citecolor=cite}
\NewCommandCopy{\oldtoc}{\tableofcontents}
\renewcommand{\tableofcontents}{\begingroup\hypersetup{hidelinks}\oldtoc\endgroup}
\begin{document}

\maketitle
\begin{abstract}
We investigate parametric resonance in oscillator networks subjected to periodically 
time-varying oscillations in the edge strengths. Such models are inspired by the well-known 
parametric resonance phenomena for single oscillators, as well as the potential rich 
phenomenology when such parametric excitations are present in a variety of applications 
like deep brain stimulation, AC power transmission networks, as well as vehicular flocking 
formations. We consider cases where a single edge, a subgraph, or the entire network is subjected
to forcing, and in each case, we characterize an interesting interplay between the parametric resonance 
modes and the eigenvalues/vectors of the graph Laplacian. Our analysis is based on a novel treatment 
of multiple-scale perturbation analysis that we develop for the underlying high-dimensional 
dynamic equations. 
\end{abstract}

\section{Introduction and Problem Formulation}
\label{sec:problem-statement}

Parametric resonance is a ubiquitous phenomenon occurring at a wide range of scales, 
from undesirable roll-motion of large ships~\cite{fossen2011parametric}, to actuated  
 micro-cantilever systems~\cite{moran2013review} where parametric resonance is used for 
 high sensitivity mass measurements. The distinction between this
phenomenon and  the more familiar harmonic resonance is that the periodically 
oscillating signal is a parameter in the system's dynamics, rather than an additive input as 
is the case with harmonic resonance. 

The simplest mathematical model of parametric resonance is the Mathieu or Hill 
equation\footnote{When $f$ is a pure sinusoid, e.g. $f(t)=a \cos(\omega t)$, it is referred to as 
	 the Mathieu equation, while it is termed Hill's equation when $f$ has higher harmonics representing 
	 a more general periodic signal. The qualitative behavior is similar in either case.}~\cite{nayfeh1995,magnus2004Hill}
      \begin{align} 
            	\ddot{x}(t) + \lb k + f(t) \romn \rb x(t) &= 0 						   \label{Mathieu1.eq}	\\
            	\Leftrightarrow~~~ 
            	\ddot{x}(t) +  k \, x(t)  & = - f(t) \, x(t)							 \label{Mathieu2.eq}
       \end{align}
where $f$ is a periodic signal which we refer to as the ``parametric excitation''. 
When $k>0$, \cref{Mathieu1.eq} can be thought of as a Mass-Spring system with a periodically time-varying 
spring constant $k+f(t)$. Alternatively, the parametric excitation term $f(t) \, x(t)$ in~\cref{Mathieu2.eq} can be thought 
of as an applied periodic force on the Mass-Spring system which is modulated (in time) by the system's 
state $x$. The unexcited system (i.e. with $f(t)=0$) has a natural frequency of $\sqrt{k}$, and when the 
frequency of the periodic, parametric excitation $f$ is tuned to certain fractions of this natural frequency, 
i.e.  at $\sqrt{k}/n, ~n=\tfrac{1}{2},1,2,3,\ldots$, the system 
above exhibits an exponentially growing instability for arbitrarily small amplitudes of $f$. Thus in 
contrast to harmonic resonance, parametric resonance occurs at several  frequencies, and the resulting 
trajectories grow exponentially rather than linearly. In many models, there are nonlinearities 
in~\cref{Mathieu1.eq} (e.g. stiffening springs) that serve mainly to saturate those exponentially growing oscillations 
into stable limit cycles. 

When $k<0$ in~\cref{Mathieu1.eq}, this system can be considered as the linearization of the 
Kapitza pendulum~\cite{Berg_2015} undergoing vibrational stabilization. In this case, without 
 parametric excitation (i.e. $f(t)=0$), the system is unstable. For certain ranges
of amplitudes and frequencies of $f$, the system can be stabilized. This is stabilization without 
sensor feedback, commonly known as vibrational stabilization. A compelling argument
can be made~\cite{Berg_2015} that vibrational stabilization and parametric resonance are 
two different manifestations of one underlying phenomenon, namely that with parametric, 
periodic excitation, the stability of certain systems can be ``flipped'', i.e. from stable to 
unstable in the parametric resonance case, and from unstable to stable in the vibrational 
stabilization case. We note here that since  vibrational stabilization has historically been 
studied using averaging methods~\cite{bullo2002averaging,meerkov1980principle}, it is 
sometimes incorrectly assumed that it requires the use of high excitation frequencies. In fact, 
it can occur at any frequency, but the range of stabilizing amplitudes (of $f$) becoming smaller
as the frequency is reduced. One of the compelling aspects of vibrational control (as well as 
parametric resonance) is that it is a form of {\em sensorless control} which does not require 
feedback in order to change a system's stability properties, and may therefore be useful in 
setting where sensing is either unavailable or difficult to implement. 

The starting point of the present work is the fascinating observation made 
in~\cite{qin2022Vibrational,histed2009direct} regarding possible mechanisms behind 
the clinical  technique of Deep Brain Stimulation (DBS). Roughly speaking, 
the observation is that DBS essentially modulates the connection strengths between 
axons and dendrites, and that DBS is plausibly a ``network effect'' in an neural network 
whose edges have periodically varying strengths.  This is then formalized mathematically~\cite{qin2022Vibrational,nobili2023vibrational} as a
 vibrational control (through periodic edge-strength modulation) problem of an oscillator network. 
 
 In this paper, we take an alternative approach based on studying parametric resonance, rather than 
 vibrational stabilization in oscillator networks. The basic model we adapt is for a network of undamped, 
 second-order linear oscillators of the following form 

\begin{align}
\label{eq:linearized-swing-dynamics}
\ddot{\phi}(t) + L\, \phi(t) = 0, 
\end{align}
where $L$ is an $n\times n$  graph Laplacian, and 
$\phi\in\R^n$ is a vector of phases where the $i^{\text{th}}$ component is the phase of the $i^{\text{th}}$
oscillator. 

Examples of such models are the linearized ``swing dynamics'' of AC power 
networks~\cite{dorfler2014synchronization}, 2nd order models of flocking in vehicular formations~\cite{bamieh2008effect}, 
and classical models of molecular vibrations~\cite{chalopin2019universality}. A useful 
physical analogy~\cite{dorfler2014synchronization} 
for the model~(\ref{eq:linearized-swing-dynamics}) is that the $ij^{\text{th}}$ entry of the Laplacian $L$ is the coefficient of 
a spring-force coupling between the phases $\phi_i$ and $\phi_j$. 

We study parametric resonance in~(\ref{eq:linearized-swing-dynamics}) by assuming that edge strengths (namely the entries of the Laplacian $L$) are perturbed periodically at a given frequency.   More precisely,  the phase dynamics are given by
\begin{align}
\label{eq:linearized-swing-dynamics-perturbed}
\ddot{\phi} + \left( L + \epsilon \,  P  \, \cos(\omega t) \romn \right) \phi = 0,
\end{align}
where \(P\) is the graph Laplacian of a subnetwork, and \(\epsilon\) and \(\omega\) are the strength and frequency of the edge weight forcing respectively. This subnetwork $P$ can be a single edge, a proper subnetwork,  or the full network itself. 
\cref{eq:linearized-swing-dynamics-perturbed} is the higher-order analogue of Mathieu's equation~\cref{Mathieu1.eq}.  By generalizing the periodic function from a one-harmonic sinusoid to a general $\sT$-periodic function, we can also obtain the higher-order analogue of Hill's
\begin{align}
\label{eq:linearized-swing-dynamics-perturbed-as-hill-ode}
\ddot{\phi} + \left( L + \epsilon \,  P \,  f(t) \romn \right) \phi = 0, \quad f(t + \sT) = f(t). 
\end{align}
In either case, an important assumption we make in this paper is that the  the subnetwork \(P\) is perturbed at a \emph{single} frequency, or by a single scalar function \(f(t)\).

The following two extreme cases serve as a motivation for our analysis. \begin{itemize}
\item \(P = L\), the network itself.  This is the ``global'' case, where the entire network is acted on by the same parametric
	excitation. This is perhaps the case in DBS where all neural connection are subject to the same external electromagnetic
	field. 
\item \(P = E_{ik} := e_{ik} e_{ik}^{\star}\), where \(e_{ik}\) is the vector with \(1\) and \(-1\) in the \(i^{\text{th}}\) and \(k^{\text{th}}\) positions and \(0\) elsewhere.  It is the Laplacian of an \(n\)-node graph with a single edge from node \(i \to k\).  This is the case of a  ``single line'' perturbation as what might happen in an AC power transmission network where a single line 
is perturbed by say a thunderstorm. 
\end{itemize}

We believe this problem formulation is novel and maybe of interest in the application areas mentioned. In the present paper, 
we are interested in the following basic questions about~\cref{eq:linearized-swing-dynamics-perturbed}.  
\begin{itemize}
\item Does the system experience parametric resonance (instability) in the limit \(\epsilon \to 0\)?
\item How are the destabilizing (``critical'') frequencies of the external excitation  related to the graph structure, \(L\)?
\item Which parametric resonance modes are the most robust or the easiest to excite?
\item How does the choice of \(P\) affect the stability of the system?  Is perturbing certain edge weights more likely to cause parametric resonance than others?
\end{itemize}

This paper is organized as follows.
We cover the special, simpler case of "global" or full-network forcing (\(P = L\)) in \cref{sec:full-network-parametric-forcing}. The more challenging single-line case  (\(P = E_{ik}\)) is treated in \cref{sec:single-edge-parametric-forcing} using a vector-valued generalization of multiple-scale perturbation analysis developed in the appendices.  
In \cref{sec:generalization-partial-network-forcing} we extend the stability results obtained for single line perturbations to the more general case of parametric forcing of an arbitrary subset of network edges.  \cref{sec:controllability-interpretation} provides an alternate interpretation of the stability results from the perspective of system controllability.  \cref{sec:ring-network-susceptibility-analysis} applies the main stability result to the case of periodic edge perturbations of ring networks.  The simple topology of rings and periodic lattices (tori) allows us to investigate the effect of network size and connectivity on its susceptibility to parametric resonance.  Finally, \cref{sec:concl-future-work} summarizes our findings and concludes with open questions and possible generalizations of the method.

The next two subsections  gather    preliminaries on   graph Laplacians in   \cref{sec:laplacian-properties}  used in  subsequent analysis, as well as the basic stability properties of the Mathieu equation in \cref{sec:Mathieu-properties}.

\subsection{Properties of Graph Laplacians}
\label{sec:laplacian-properties}

We briefly review here basic  properties of  graph Laplacians that we use in our analysis. We consider only 
undirected graphs, and thus $L$ or the subgraph Laplacian $P$ are real symmetric, positive semi-definite 
with at least one eigenvalue at $0$ with eigenvector \(\bone\), i.e. $L \bone = 0 \bone$.

The other eigenvalues give the natural frequencies  of oscillation of this system as the square roots of the 
eigenvalues
\begin{align*}
L \,  \mathbf{v}_k & = {\scriptstyle \Omega}_k^2 \, \mathbf{v}_k, \quad  {\scriptstyle \Omega}_{k} \in \mathbb{R} \\
{\scriptstyle \Omega}_k & \ne {\scriptstyle \Omega}_m \implies \mathbf{v}_k^{*} \, \mathbf{v}_m = 0,
\end{align*}   
and the mutual orthogonality of the eigenvectors is due to the symmetry of the Laplacians. 

Also of interest to us is the Hamiltonian matrix \(\mathcal{L}\) that acts as the generator of system \eqref{eq:linearized-swing-dynamics}
\begin{align}
  \mathcal{L} := \begin{bmatrix} 0 & I \\ \text{-}L & 0 \end{bmatrix}
\end{align}
The non-zero eigenvalues and eigenvectors of \(\mathcal{L}\) are related to those of \(L\) by 
\begin{align}
L \, \mathbf{v} = {\scriptstyle \Omega}^2 \, \mathbf{v} \implies \mathcal{L} \begin{bmatrix} \mathbf{v} \\ \pm j {\scriptstyle \Omega} \mathbf{v} \end{bmatrix} = \pm j {\scriptstyle \Omega} \begin{bmatrix} \mathbf{v}  \\ \pm j {\scriptstyle \Omega} \mathbf{v} \end{bmatrix} \label{eq:graph-laplacian-as-generator}
\end{align}
Since \(L\) has non-negative real eigenvalues, all eigenvalues of \(\mathcal{L}\) are on the imaginary axis.  Further, each zero eigenvalue of \(L\) corresponds to a Jordan block of \(\mathcal{L}\) of size \(2\) with eigenvalue zero.  For a network with one connected component, the corresponding eigenvector is $\bone$, the vector of \(1\)s.

As a result, the modes of \cref{eq:linearized-swing-dynamics-perturbed} corresponding to the zero eigenvalue are always unstable, even at \(\epsilon = 0\).  However, these modes are physically unimportant: they correspond to drift in the average value of the state \(\phi\) and do not otherwise influence the stability of the system.

\subsection{Stability of the Mathieu Equation}
\label{sec:Mathieu-properties}

Consider the  Mathieu equation 
\be
	\ddot{x}(t) \,+\, \lb \omega_\rmn^2 + \epsilon \, \cos(\omega t) \romn \rb x(t) ~=~ 0, 
  \label{Mathieu_norm.eq}
\ee
where $\omega_\rmn=\sqrt{k}$ is the unexcited (i.e. $\epsilon=0$) natural frequency of the system. 
For any amplitude $\epsilon$ and frequency $\omega$ of the parametric excitation, Equation~\cref{Mathieu_norm.eq}
is a linear periodic system (with period $2\pi/\omega$), and therefore 
its stability properties can be determined using the eigenvalues of its monodromy matrix, which in general must be 
computed numerically. Figure~\ref{fig:mathieu-instability-rescaled} shows the regions of stability/instability in the
$(\epsilon,\omega)$ plane. The 
grey-shaded areas are the unstable regions, while the remainder are the (neutrally) stable regions. 

The unstable regions ``touch'' the $\epsilon=0$ axis at discrete multiples of $\omega$, $\tfrac{\omega}{\omega_n} = 2, 1, \tfrac{1}{2}, \tfrac{1}{3}, \ldots$.  These represent excitation frequencies for which the system~\eqref{Mathieu_norm.eq} can be destabilized with 
arbitrarily small excitation amplitudes. They are
 referred to as ``Arnold tongues'',  and can also be found using perturbation methods. 
The behavior of each Arnold tongue around $\epsilon\approx 0$ depends on its index.  The first tongue at 
$\omega=2\omega_\rmn$ expands linearly for small epsilon, i.e. the stability boundary behaves like 
$\omega(\epsilon) =2 \omega_\rmn \pm a_1 \epsilon + O(\epsilon^2)$, while the stability boundary at the 
second tongue behaves like $\omega(\epsilon) = \omega_\rmn  \pm a_2 \epsilon^2 + O(\epsilon^3)$. More generally, 
the ``width'' of the $i^{\text{th}}$ tongue behaves like $\epsilon^n$, which makes each successive tongue ``thinner''. This 
has important implications for observing parametric resonance phenomena in systems with small  damping, 
as damping will effectively ``lift'' the tongues away from the $\epsilon=0$ axis as shown in 
\cref{fig:mathieu-instability-rescaled} ({\em right}). For the higher order tongues (i.e. $i\geq 2$), 
this effect makes the threshold of $\epsilon$ where instability is observed much higher than $\epsilon=0$, and consequently, only the first tongue is observed in applications where there is any amount of damping.

\begin{figure}[t]
\centering
	\includegraphics[width=0.45\columnwidth]{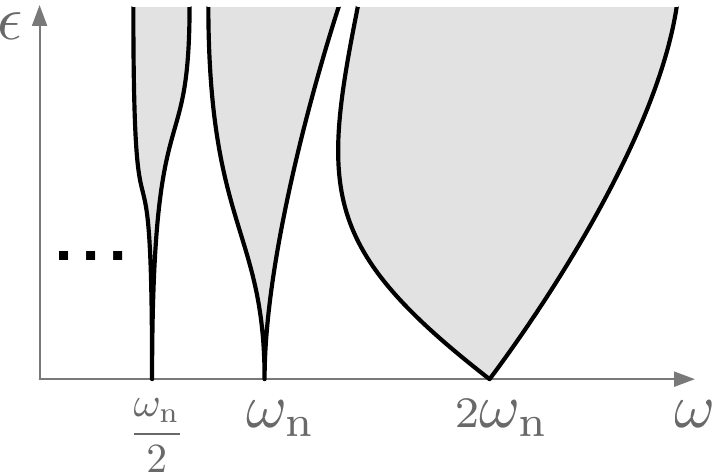} 
	\quad
	\includegraphics[width=0.45\columnwidth]{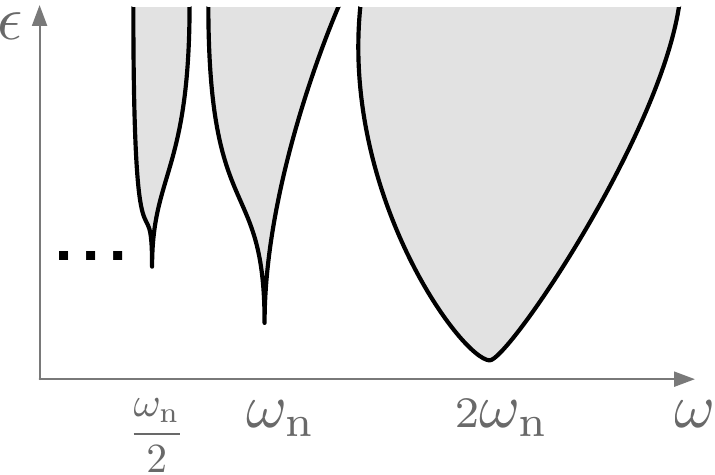}

	\caption{\label{fig:mathieu-instability-rescaled} The stability diagram for the undamped 
		Mathieu equation  ({\em left}),  and for the Mathieu equation with slight damping ({\em right}), where the 
		gray-shaded regions are the unstable regions in the $(\epsilon,\omega)$ plane of the parametric  
		excitation amplitude $\epsilon$ and frequency $\omega$.  
		In the undamped case, the system is unstable for arbitrarily small excitation amplitude $\epsilon$ at 
		excitation frequencies that are $2,1,1/2,1/3,...$ times the natural frequency $\omega_\rmn$. 
		The shaded areas near $\omega\approx \omega_\rmn/k, ~k=1/2,1,2,3,4,..$
		 are called the ``Arnold tongues'', which get progressively thinner with increasing $k$. Thus with slight damping, 
		 only the first tongue at $\omega\approx 2\omega_\rmn$ represents an instability phenomenon that can be 
		 triggered  with small 
		 excitation amplitude $\epsilon$. Only the first three tongues are shown in the diagrams above.  }
\end{figure}

The stability diagram for the more general Hill equation is qualitatively similar to that of ~\cref{fig:mathieu-instability-rescaled}. The shapes of the instability regions in that case are generally different, and depend on the details of the periodic function $f$. However, the asymptotic behavior (around $\epsilon\approx 0$) of the Arnold tongues is the same as shown in the figure.

\section{Full-network parametric forcing}
\label{sec:full-network-parametric-forcing}

We now consider the case where $P=L$ in model~\eqref{eq:linearized-swing-dynamics-perturbed}. This represents a scenario where all the edges of the network are subjected to the same oscillatory perturbation simultaneously. This case
is the simplest to analyze since 
\cref{eq:linearized-swing-dynamics-perturbed} becomes
\begin{align}
\ddot{\phi} + \left( \romn L + \epsilon \, L \, \cos(\omega t) \right) \phi  &= 0 , 		\nonumber	\\ 
\Leftrightarrow \hspace{2em} 
\ddot{\phi} +\left( \romn 1 + \epsilon \cos(\omega t) \right)   L \,  \phi & = 0.
\label{eq:linearized-swing-dynamics-perturbed-full-network}
\end{align}
Since \(L\) has a full set of independent eigenvectors, this system is diagonalizable 
using $L\, V = V\, \Omega^2$ where $V$ is the matrix whose columns are the eigenvectors, and $\Omega^2$ is a diagonal 
matrix made up of the corresponding eigenvalues $\lcb \sOmega^2_k \rcb$. 
The result is a collection of uncoupled Mathieu's equations:
\begin{align}
& L \, V = V \,  \Omega^{2} \implies \nonumber\\
& \left( V^{\text{-}1} \ddot{\phi} \right) + \lb \romn 1 + \epsilon \cos(\omega t)\rb   \Omega^{2} \left( V^{\text{-}1} \phi \right) = 0. \label{eq:lsd-perturbed-full-network-decoupledtm1}
\end{align}
Defining \(p(t) := V^{\text{-}1} \ddot{\phi}(t)\), the \(j^{\text{th}}\) such system is
\begin{align}
\label{eq:lsd-perturbed-full-network-decoupled-2}
\ddot{p}_j(t) + {\scriptstyle \Omega}_j^2 \lb 1 + \epsilon \cos(\omega t) \romn\rb   p_j(t) = 0.
\end{align}
This is the scalar Mathieu equation for an oscillator with unperturbed ``natural frequency'' \(\sOmega_j\) and parametric forcing of amplitude \({\scriptstyle \Omega}_j^2  \epsilon\) at frequency \(\omega\).  
For each $j$, the 
 stability properties of the system~\eqref{eq:lsd-perturbed-full-network-decoupled-2} 
 are the same as the  corresponding Mathieu equation~\eqref{Mathieu_norm.eq} as shown in~\cref{fig:mathieu-instability-rescaled}.

\begin{figure}[t]
\centering
\includegraphics[width=0.8\columnwidth]{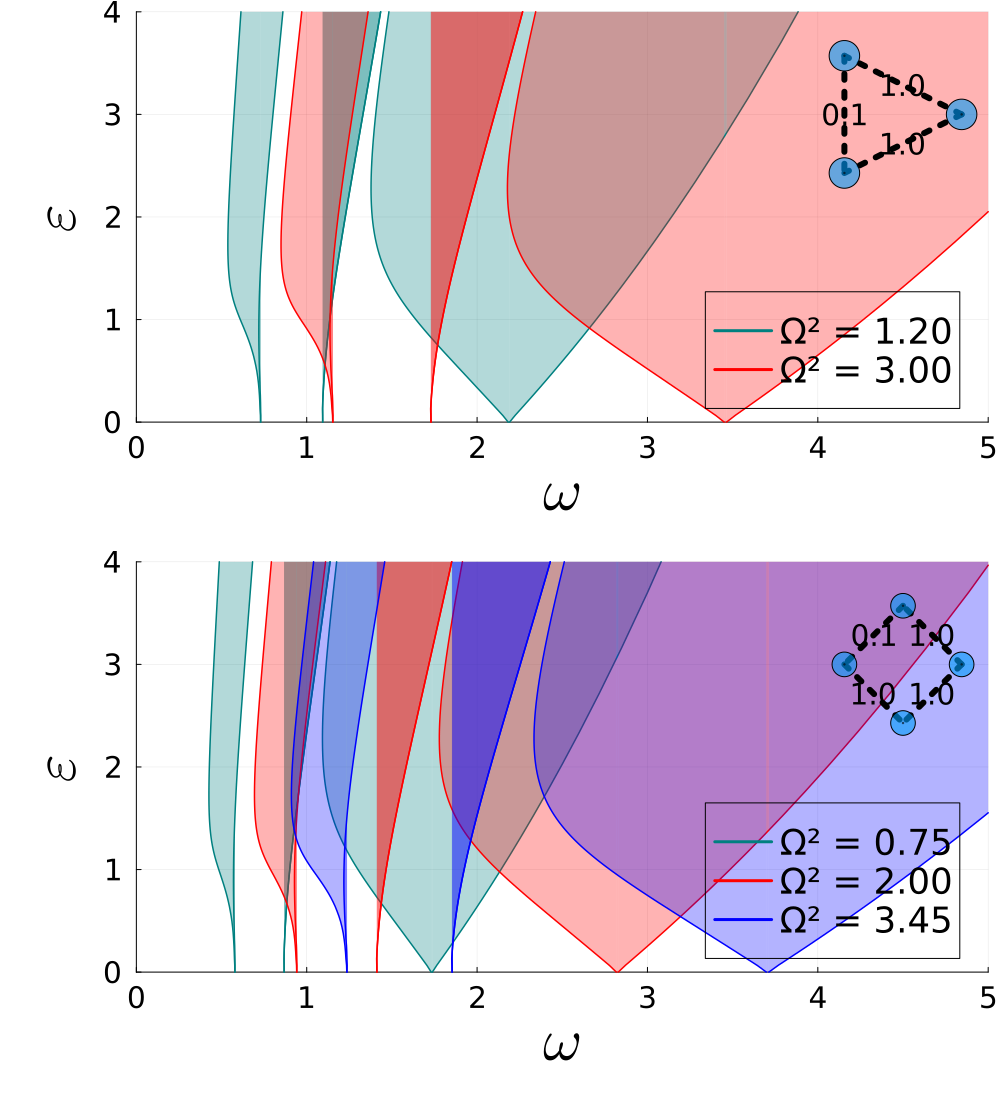}
\caption{\label{}\label{fig:mathieu-instability-full-network-forcing-examples} Stability diagrams for networked second order oscillators corresponding connected in a triangle (top) and square (bottom) layout with specified nominal edge weights.  The stability diagrams are overlapping copies of scaled versions of Figure \ref{fig:mathieu-instability-rescaled}, one per mode of the system.  The region of instability is the union of the unstable regions for each mode.}
\end{figure}

The stable region in \((\omega, \epsilon)\) space for~(\ref{eq:linearized-swing-dynamics-perturbed-full-network}) is thus the intersection of the stable regions for all the modes \(p_j,\ j = 1, \dots n\).  
\cref{fig:mathieu-instability-full-network-forcing-examples} 
shows such stability diagrams for two examples with 3 and 4 node networks. The basic pattern is as follows. 
For an 
$n$-node connected network, the Laplacian has $n-1$ non-zero eigenvalues, and therefore by~(\ref{eq:lsd-perturbed-full-network-decoupled-2}) there are $n-1$ tongues of each type. More precisely, there are $n-1$ first order tongues at 
excitation frequencies of 
\[
	\omega ~=~ 2 \, \sOmega_j, 
	\hspace{2em} j=1,\ldots,n-1, 
\]
with $\lcb \sOmega_j^2 \rcb$ being the non-zero eigenvalues of the Laplacian $L$. Similarly, there are $n-1$ second-order
tongues at excitation frequencies $\omega = \sOmega_j$, and so on. 

\cref{fig:mathieu-instability-full-network-forcing-examples} shows the first, second, and third-order tongues for two 
networks of 3 and 4 nodes respectively. For example, the 3-node network has two tongues of each type whose locations
are determined by the eigenvalues of the Laplacian. As mentioned in Section~\ref{sec:Mathieu-properties}, only the 
first-order tongues are phenomenologically relevant in systems with any amount of damping. 

\section{Subnetwork Parametric Excitation}
\label{sec:single-edge-parametric-forcing}

In this section we  consider periodic perturbation of the weight of a portion of the network. We begin with the case of single-edge 
forcing as the analysis is very similar to the more general case of a subnetwork as shown in \cref{sec:generalization-partial-network-forcing}.

In the single-edge forcing case, 
 \(P = e_{ik} e_{ik}^{\star} =: E_{ik}\), the Laplacian of the \(n\)-node graph with a single edge from node \(i \to k\).  Then system \eqref{eq:linearized-swing-dynamics-perturbed} becomes
\begin{align}
\label{eq:linearized-swing-dynamics-perturbed-one-edge}
\ddot{\phi} + \left( L + \epsilon\, E_{ik} \cos(\omega t) \romn \right) \phi = 0
\end{align}

Unlike the case of the full network forcing, we cannot treat this system as a set of uncoupled second-order oscillators
since the pair $L,E_{ik}$ will typically not be simultaneously diagonalizable. Instead, we will find its critical parametric forcing frequencies and stability behavior using a perturbation expansion around \(\epsilon = 0\). 
This analysis is quite delicate, and is presented in Section~\ref{sec:regular-perturbation-analysis}. It relies 
on a vector-valued multiple scale perturbation analysis that we develop in the appendices. This is a generalizations of the standard scalar-valued multiple scale analysis~\cite{nayfeh1995,bender1999Advanced} and we believe is one contribution of the present paper.

We now summarize the main  stability result of this section before presenting the perturbation technique. 
For the reasons outlined earlier, we only consider the first-order instability tongues since they are the most
phenomenologically important. 

Our main result is summarized next. 
\begin{itemize} 
	\item All first-order tongues of the stability diagram of~(\ref{eq:linearized-swing-dynamics-perturbed-one-edge})
		occur at the following excitation frequencies 
		\be
			\omega_{lm} =  \sOmega_l + \sOmega_m , 
				\hstm \sOmega_l,\sOmega_m \in   \sqrt{ {\rm eigs}(L)- \{0\} } ,
		  \label{tongue_freq.eq}
		\ee
		where $\sOmega_l, \sOmega_m$ are any two non-zero natural frequencies. 
 	\item At each $\omega_{lm}$, the behavior of the stability boundary of that tongue is 
		\begin{align}  
			\omega(\epsilon) &= \omega_{lm} \pm a_{lm} \epsilon + O(\epsilon^2) 	 \nonumber 	\\ 
			\mbox{where}\hstm 
			a_{lm} &= \frac{\abs{\mathbf{v}_m^{*} \,  e_{ik} e_{ik}^{\star} \,  \mathbf{v_l}}}
						{2 \left( {\scriptstyle \Omega}_l \, {\scriptstyle \Omega}_m \right)^{1/2}}
		  \label{a_widths.eq}
		\end{align}
\end{itemize} 

The first statement~\req{tongue_freq} says that all first order tongues occur at sums of the natural frequencies 
of the unperturbed system. Those in turn are given by the square roots of the non-zero eigenvalues of the network Laplacian. Since there 
are $n-1$ of those, then the number of first-order instability tongues is $\lb (n-1)^2+(n-1)\rb/2$ with possible 
repetitions if some natural frequencies are repeated. 

The quantity $a_{lm}$~\req{a_widths} can be visualized as the ``width'' of the $lm^{\text{th}}$ tongue as illustrated in 
Figure~\ref{tongue_angles.fig}. This quantity can be used as a proxy for the susceptibility of the system to 
instability when excited near that frequency since it captures (up to first order) the width of the 
excitation-frequency range that causes a parametric instability. We note that this number depends not 
only on the eigenvalues $\sOmega_l^2,\sOmega_m^2$ of $L$, but on the corresponding eigenvectors 
$\mathbf{v}_l,\mathbf{v}_m$ as well. Thus first-order tongues in this case will typically come with a large 
variety of widths, some of which may actually be zero.

\begin{figure}[t]
	\centering
	\includegraphics[width=0.9\columnwidth]{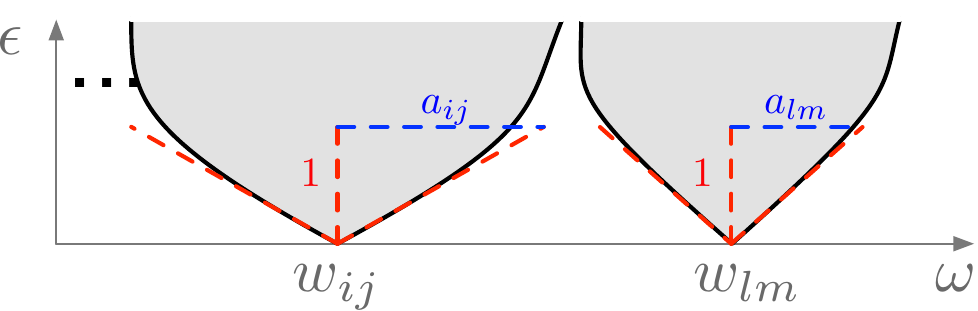}
	
	\caption{Depiction of the first-order tongues of the single-edge excitation 
		model~(\ref{eq:linearized-swing-dynamics-perturbed-one-edge}). All possible locations of 
		such tongues are given by sums  $\omega_{lm}= \sOmega_l + \sOmega_m$ of natural frequencies 
		$\sOmega_k \in \sqrt{{\rm eigs}(L)}$ of the unperturbed system, which in turn are determined by the 
		eigenvalues of the network Laplacian $L$.  The ``widths'' $a_{lm}$  
		of each tongue are determined  by 
		both eigenvalues and eigenvectors of $L$ by~\req{a_widths}. 
		} 
	\label{tongue_angles.fig} 
\end{figure}

\subsection{Perturbation Analysis}
\label{sec:regular-perturbation-analysis}

We assume that \(\phi(t; \epsilon)\) is analytic in \(\epsilon\), so we may write
\begin{align}
\label{eq:swing-dynamics-2nd-order-regular-perturbation}
\phi(t; \epsilon) & =: \phi_0(t) + \epsilon \phi_1(t) + \epsilon^2 \phi_2(t) + \mathcal{O}(\epsilon^3),
\end{align}
where \(\phi_k(t)\) are differentiable functions of time.

Applying this expansion to system \eqref{eq:linearized-swing-dynamics-perturbed-one-edge} and collecting powers of \(\epsilon\), we get the triangular system
\begin{align}
\label{eq:linearized-swing-dynamics-regular-perturbation-equations}
\mathcal{O}(\epsilon^0):\  & \ddot{\phi}_0 + L \phi_0 = 0\\
\mathcal{O}(\epsilon^1):\  & \ddot{\phi}_1 + L \phi_1 = \text{-} \cos(\omega t) E_{ik} \phi_0 \\
\vdots \nonumber\\
\mathcal{O}(\epsilon^{n}):\  & \ddot{\phi}_n + L \phi_n  = \text{-} \cos(\omega t) E_{ik} \phi_{n \text{-} 1}.
\end{align}
\(\phi_k\) at each order of \(\epsilon\) behaves as a coupled harmonic oscillator with additive input \(\phi_{k\text{-}1}\).  \(\phi_0(t)\) is the solution to the unperturbed, linearized oscillator dynamics~\eqref{eq:linearized-swing-dynamics}:
\begin{align}
\label{eq:swing-dynamics-2nd-order-regular-perturbation-Oeps0-sol}
\phi_0(t) =  e^{ j \sqrt{L} t } (V \mathbf{c}) + e^{ \text{-}j \sqrt{L} t } (V \mathbf{c})^{\star}
\end{align}
for some \(\mathbf{c} \in \mathbb{C}^n\).

To analyze the \(\mathcal{O}(\epsilon^{1})\) term, it will be useful to work in the basis of the eigenvectors of \(L\).  With \(L V = V \Omega^2\) (\cref{sec:laplacian-properties}), we can write
\begin{align}
V^{\text{-}1} \phi_0(t) = e^{j \Omega t} \mathbf{c} + e^{\text{-}j \Omega t} \mathbf{c}^{\star} \label{eq:1st-order--sol-Oeps1-spectral}
\end{align}
for appropriate initial conditions \(\mathbf{c}\).  The \(\mathcal{O}(\epsilon^1)\) system is
\begin{align}
& \left( V^{\text{-}1} \ddot{\phi}_1 \right) + \Omega^2 \left( V^{\text{-}1} \phi_1 \right) = - V^{\text{-}1} E_{ik} V \cos(\omega t) \left( V^{\text{-}1} \phi_0 \right) \nonumber\\
& = -\tfrac{1}{2} V^{\text{-}1} E_{ik} V \cos(\omega t) \left( e^{j \Omega t} \mathbf{c} + e^{\text{-} j \Omega t} \mathbf{c}^{\star} \right) \nonumber\\
& = -\tfrac{1}{2} V^{\text{-}1} E_{ik} V \left( e^{j (\Omega \pm \omega I) t} \mathbf{c} + e^{\text{-} j (\Omega \pm \omega I) t} \mathbf{c}^{\star} \right)\label{eq:swing-dynamics-2nd-order-Oeps1-spectral-projected}
\end{align}
Each forcing term in system~\eqref{eq:swing-dynamics-2nd-order-Oeps1-spectral-projected} is thus of the form
\begin{align}
- \tfrac{1}{2} \mathbf{v}_m^{\star} E_{ik} \mathbf{v}_l \left( e^{j (\Omega_l \pm \omega) t} c_l + e^{j (\Omega_l \pm \omega) t} c^{\star}_l \right) \label{eq:swing-dynamics-2nd-order-Oeps1-spectral-projected-forcing-term}
\end{align}
where \(l,m = 1,2, \dots, n\).  System \eqref{eq:swing-dynamics-2nd-order-Oeps1-spectral-projected} is neutrally stable when no natural frequency \(\pm j {\scriptstyle \Omega}\) equals any frequency of the forcing \(\pm j {\scriptstyle \Omega} \pm j \omega \).  Indeed, if any of the forcing frequencies equals one of these modal frequencies, so that \(j {\scriptstyle \Omega}_m = j  {\scriptstyle \Omega}_l +  j \omega  \), say, the corresponding mode of the solution experiences harmonic resonance and is unstable.  In the parlance of perturbation theory, the corresponding forcing term is \emph{secular}.  Additional care has to be taken in this line of reasoning when the eigenvalues of \(L\) have multiplicity \(> 1\).  This case is handled in \cref{sec:multi-scale-method-with-multiplicity}.

Consequently the \(\mathcal{O}(\epsilon)\) term of the solution \eqref{eq:swing-dynamics-2nd-order-regular-perturbation}  is unstable when

\begin{align}
\label{eq:basic-instability-criterion}
 \pm {\scriptstyle \Omega}_m = \pm {\scriptstyle \Omega}_l \pm \omega
\end{align} 
for any \(l, m = 1, 2, \dots, n\).

By assumption, \(\omega > 0\), and all the eigenvalues \({\scriptstyle \Omega}\) are non-negative.  So this list of conditions reduces to the following two:
\begin{align}
\omega & = \abs{{\scriptstyle \Omega}_m + {\scriptstyle \Omega}_l} \label{eq:basic-instability-criterion-filtered-plus} \\
\omega & = \abs{{\scriptstyle \Omega}_m - {\scriptstyle \Omega}_l} \label{eq:basic-instability-criterion-filtered-minus} \\
l, m & = 1, 2, \dots, n \nonumber
\end{align}
The critical parametric forcing frequency, if it exists, must be the sum or difference of the square roots of the eigenvalues of the graph Laplacian.  This condition is subject to the validity of the regular perturbation analysis, and also requires that the coefficient of the relevant additive forcing term (\(\mathbf{v}_m^{\star} E_{ik} \mathbf{v}_l\)) corresponding to the eigenvalues \({\scriptstyle \Omega}_l\), \({\scriptstyle \Omega}_m\) of \(L\) be non-zero (\cref{eq:swing-dynamics-2nd-order-Oeps1-spectral-projected-forcing-term}).  This requirement is investigated in more detail in \cref{sec:controllability-interpretation}.

The boundedness of the \(\mathcal{O}(\epsilon^2)\) and higher order terms in the ansatz~\eqref{eq:swing-dynamics-2nd-order-regular-perturbation}  can be investigated similarly.  Assuming \(\phi_1(t)\) is bounded, it has the form
\begin{align}
\label{eq:swing-dynamics-2nd-order-regular-perturbation-Oeps1-sol}
V^{\text{-}1} \phi_1(t) = & e^{j \Omega t} \mathbf{c}_0 + e^{j (\Omega \pm \omega I)t} \mathbf{c}_1 + e^{\text{-}j (\Omega \pm \omega I) t} \mathbf{c}_2  \nonumber\\
& + \text{complex conjugates,}
\end{align}
where \(\mathbf{c}_{*}\) are constants that depend on the initial conditions.  Then \(\phi_2(t)\) is the solution of the $\mathcal{O}(\epsilon^2)$ system
\begin{align}
V^{\text{-}1} \ddot{\phi}_2 + \Omega^2\, V^{\text{-}1} \phi_2 = - \tfrac{V^{\text{-}1} E_{ik} V}{2} \left( e^{j \omega t} + e^{-j \omega t} \right) V^{\text{-}1} \phi_1
\end{align}
which contains additive forcing at frequencies \({\scriptstyle \Omega} \pm \omega\), as before, as well as at \({\scriptstyle \Omega} \pm 2 \omega\).  Thus the \(\phi_2(t)\) system experiences harmonic forcing whenever any of these frequencies equals a natural frequency of the system, or
\begin{align*}
 \pm j {\scriptstyle \Omega}_m & = \pm j {\scriptstyle \Omega}_l \pm j \omega, \text{ or} \\
 \pm j {\scriptstyle \Omega}_m & = \pm j {\scriptstyle \Omega}_l \pm 2 j \omega
\end{align*}
for any \(l, m = 1, 2, \dots, n\).  The first of these conditions does not hold since, by assumption, \(\phi_1(t)\) is stable.  As before (\cref{eq:basic-instability-criterion-filtered-plus} and \cref{eq:basic-instability-criterion-filtered-minus}), this list of conditions reduces to \(\omega = \abs{{\scriptstyle \Omega}_m \pm {\scriptstyle \Omega}_l}/2\).  By a similar argument, we can show that the \(\mathcal{O}(\epsilon^{n})\) term \(\theta_n(t)\) is unstable when
\begin{align}
\omega & = \tfrac{1}{n} \abs{{\scriptstyle \Omega}_m + {\scriptstyle \Omega}_l} \label{eq:basic-instability-criterion-filtered-Oeps2-plus} \\
\omega & = \tfrac{1}{n} \abs{{\scriptstyle \Omega}_m - {\scriptstyle \Omega}_l} \label{eq:basic-instability-criterion-filtered-Oeps2-minus}
\end{align}

\cref{fig:basic_instability_criterion_unfiltered} compares (to \(\mathcal{O}(\epsilon)\)) the numerically computed stability diagram for networked oscillators with our predictions for the critical parametric forcing frequencies.  We observe that:
\begin{itemize}
\item While all the critical frequencies (as \(\epsilon \to 0\)) to first order are accounted for, criteria~\eqref{eq:basic-instability-criterion-filtered-plus} and \eqref{eq:basic-instability-criterion-filtered-minus} generate ``false positives'' as well.
\item The regular perturbation analysis gives us no information on the slopes or widths of the tongues.  Thus we cannot order by severity (or excitability of parametric resonance) either the critical frequencies or choices of the edge to be perturbed.
\end{itemize}

To refine the instability conditions, we apply a multiple time-scale perturbation analysis to our original system \eqref{eq:linearized-swing-dynamics-perturbed-one-edge}.  This allows us to investigate the behavior of the system in the vicinity of the critical frequencies in \((\omega, \epsilon)\) parameter space, resolve the ``false positives'' and obtain a measure of the instability at each critical frequency.

\begin{figure}[t]
\centering
\includegraphics[width=0.8\columnwidth]{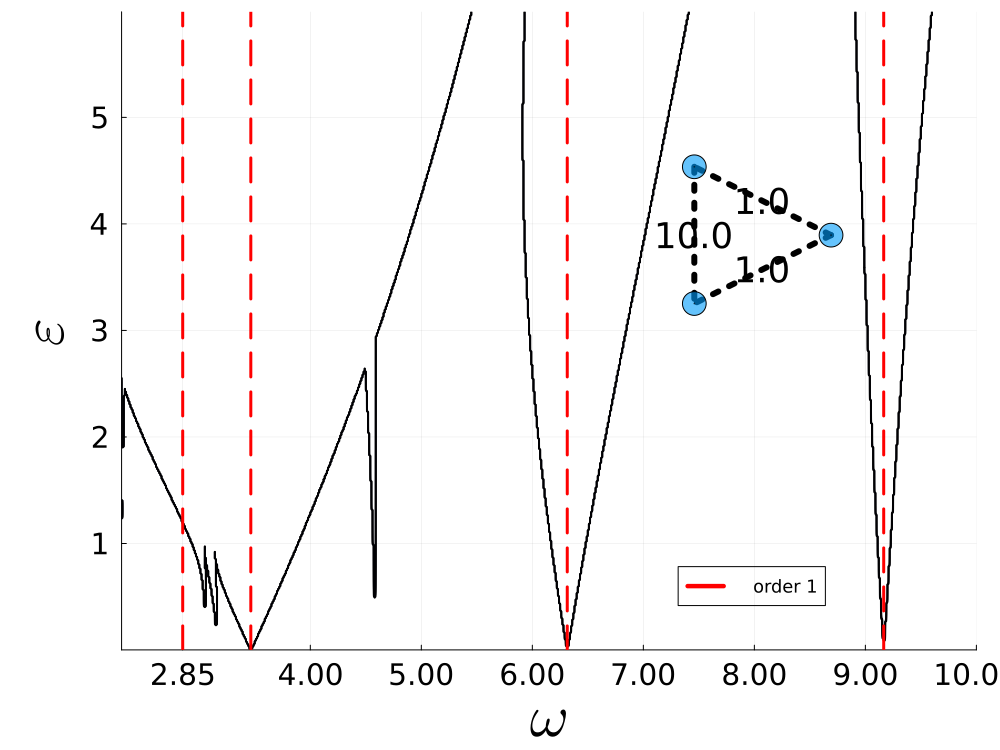}
\caption{\label{}\label{fig:basic_instability_criterion_unfiltered} The numerically computed stability diagram for a triangular graph with unequal edge weights.  The overlaid red lines are the critical parametric forcing frequencies predicted (to first order) by the regular perturbation analysis of Section \ref{sec:regular-perturbation-analysis}.  The narrow tongues at \(\omega \approx 2.9\) and \(\omega \approx 4.7\) correspond to the predictions from higher order perturbations (\cref{eq:basic-instability-criterion-filtered-Oeps2-plus} and \cref{eq:basic-instability-criterion-filtered-Oeps2-minus}).  We observe that while the Arnold tongue positions are predicted correctly by the perturbation analysis, the prediction of \(2.85\) is a false positive.  The perturbation analysis is refined using a multiple time-scale method in \cref{sec:vector-valued-multiple-scale-perturbation-analysis}.}
\end{figure}

We do this in the following sequence:
\begin{enumerate}
\item Investigate the stability of the system along straight lines in the \((\omega, \epsilon)\) parameter space through the critical frequencies.  Each line is determined by the \(\omega\)-intercept \( \omega_0 \) and the slope \(\left. \tfrac{\mathrm{d}\omega}{\mathrm{d}\epsilon}\right\rvert_{\epsilon = 0} =: a \).
\item Expand \(\phi(t)\) along the fast and slow time scales \(t\) and \(\epsilon t\).  This decomposition is non-unique.
\item System stability requires both the fast and slow time dynamics to be stable.  Enforce this condition to give us both (i) a suitable decomposition of the dynamics and (ii) stability criteria for the \(\omega\) intercept and \( a \).
\end{enumerate}

The details of the multiple scale perturbation analysis are presented in \cref{sec:vector-valued-multiple-scale-perturbation-analysis}.  We summarize the results here:

\begin{theorem}[Instability criteria]
Consider the curve in the \((\omega, \epsilon)\) parameter plane given by \(\omega(\epsilon) = \omega_0 + a \epsilon + \mathcal{O}(\epsilon^2)\).  Let \({\scriptstyle \Omega}_m^2\) and \(\mathbf{v}_m\) be the eigenvalues and eigenvectors of the graph Laplacian \(L\), for \(m = 1, 2, \dots, n\).  In the limit \(\epsilon \to 0\), system~\eqref{eq:linearized-swing-dynamics-perturbed-one-edge} is unstable at \((\omega(\epsilon), \epsilon)\) if
\begin{align}
& \omega_0 = {\scriptstyle \Omega}_l + {\scriptstyle \Omega}_m,\ l, m \in 1, 2, \dots, n \nonumber\\
& \abs{a} < \frac{\abs{\mathbf{v}_m^{\star} E_{ik} \mathbf{v}_l}}{2 \left( {\scriptstyle \Omega}_l {\scriptstyle \Omega}_m \right)^{1/2}}. \label{eq:multiscale-analysis-full-criteria}
\end{align}
\end{theorem}

We interpret the result as follows:
\begin{itemize}
\item The stability criteria obtained from the regular perturbation analysis suggest that both \(\pm \omega_0 = {\scriptstyle \Omega}_l - {\scriptstyle \Omega}_m\) and \(\omega_0 = {\scriptstyle \Omega}_l + {\scriptstyle \Omega}_m\) cause instability.  However the method of multiple time-scales rules out the former.  This is the first source of false positives in the results suggested by the regular perturbation analysis.
\item \( a = 0\) corresponds to a vertical line in the \((\omega, \epsilon)\) parametric plot.  This implies that to \(\mathcal{O}(\epsilon)\), the stability boundary has a symmetric V-shape, where the critical angle from the vertical to either side is \(\pm \tan^{-1} \left( \tfrac{\abs{\mathbf{v}_m^{\star} E_{ik} \mathbf{v}_l}}{ 2 \left( {\scriptstyle \Omega}_l {\scriptstyle \Omega}_m \right)^{1/2} } \right)\).
\item If \(\mathbf{v}_m^{\star} E_{ik} \mathbf{v}_l = 0\), then the system is stable around \((\omega_0, 0)\).  This is the second source of false positives in the regular perturbation analysis.  It is also consistent with the result of the controllability analysis (\cref{sec:controllability-interpretation}), which posits the same criterion for instability.
\end{itemize}
\subsection{Comparison with Numerical Calculations}
\cref{fig:gtri-10-edge-12-slope-comparison} compares the approximation~\eqref{eq:multiscale-analysis-full-criteria} to the results of numerical computation for the stability of a test graph.  We observe that the largest Arnold tongue structures correspond to the critical values of \(\omega\) predicted by the \(\mathcal{O}(\epsilon)\) term in the perturbation analysis.  The slopes of the tongues as \(\epsilon \to 0\) are captured accurately by the multiple time-scale analysis above.

\begin{figure}[t]
\centering
\includegraphics[width=0.8\columnwidth]{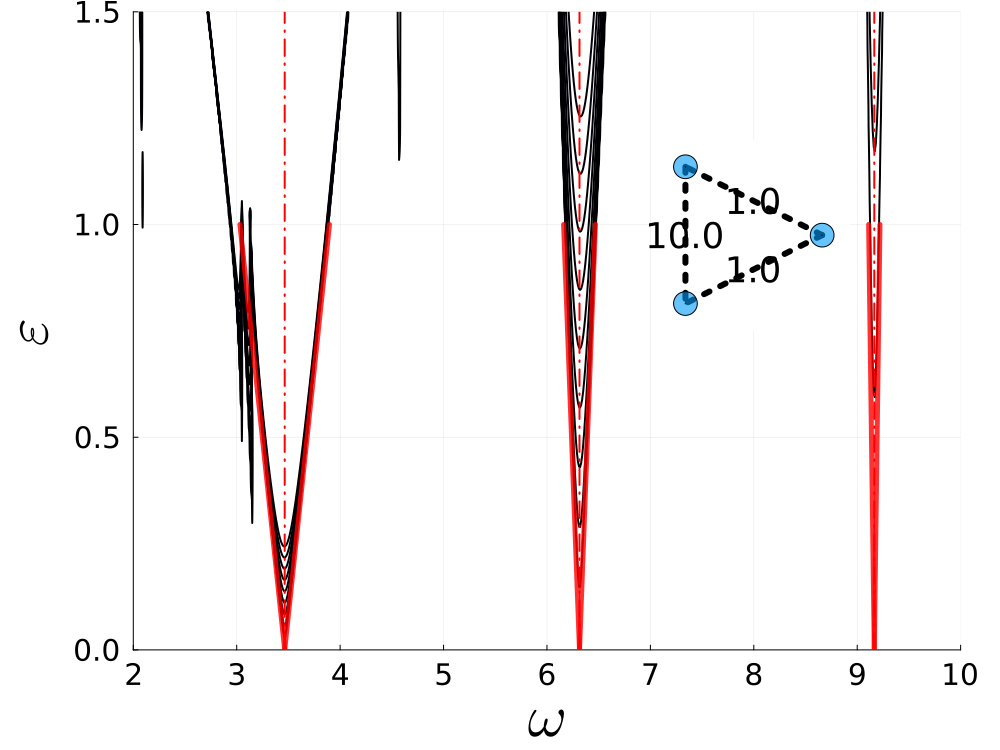}
\caption{\label{fig:gtri-10-edge-12-slope-comparison}A comparison of the numerically computed stability diagram with the stability estimates from perturbation analysis for swing dynamics on triangular graph with one time-periodically perturbed edge.  The wide Arnold tongue positions in \((\omega, \epsilon)\) space (black curves) correspond to first order perturbation analysis and are predicted correctly (dashed red lines).  The multiple time-scale analysis from section \ref{sec:vector-valued-multiple-scale-perturbation-analysis} eliminates the false positives observed in figure \ref{fig:basic_instability_criterion_unfiltered}.  Additionally, the estimate of the slope of the Arnold tongues as \(\epsilon \to 0\) computed using \eqref{eq:k-to-omega-transformation-2} (solid red lines) match the numerical results closely, and thus provide a simple way of predicting the degree of instability of each critical parametric forcing frequency \(\omega\) from the properties of the graph Laplacian.  The narrower, partially resolved Arnold tongues correspond to higher order terms in the perturbation expansion \eqref{eq:swing-dynamics-2nd-order-regular-perturbation}.}
\end{figure}

\subsection{Controllability Interpretation}
\label{sec:controllability-interpretation}

We can understand the stability criterion \eqref{eq:multiscale-analysis-full-criteria} using a controllability argument.  The oscillator dynamics with single edge forcing are given by \cref{eq:linearized-swing-dynamics-perturbed-one-edge}, or
\begin{align*}
\ddot{\phi} + \left( L + \epsilon E_{ik} \cos(\omega t) \right) \phi = 0.
\end{align*}
Following the perturbation expansion \eqref{eq:swing-dynamics-2nd-order-regular-perturbation} and solution~\eqref{eq:swing-dynamics-1st-order-multiple-scale-sol-Oeps1}, we have, up to \(\mathcal{O}(\epsilon)\)
\begin{align*}
\phi_0(t) & = V \left( e^{j \Omega t} \mathbf{c} + e^{\text{-}j \Omega t} \mathbf{c}^{\star} \right)
\end{align*}
where \(\phi_0(t)\) is the solution to the unperturbed, linearized swing equations, and appears as an additive input in the equation for the \(\mathcal{O}(\epsilon)\) term \(\phi_1(t)\):
\begin{align*}
\ddot{\phi}_1(t) & = -L \phi_1(t) - E_{ik} \cos(\omega t) \phi_0(t) \\
& = -L \phi_1(t) - \underbrace{E_{ik} V}_{\mathcal{B}} \underbrace{e^{j (\Omega + \omega I) t} \mathbf{c}}_{u_1(t)} - \underbrace{E_{ik} V}_{\mathcal{B}} \underbrace{e^{\text{-}j (\Omega \text{-} \omega I)t}  \mathbf{c}}_{u_2(t)} \\
& \quad - \mathcal{B} u_1^{\star}(t) - \mathcal{B} u_2^{\star}(t).
\end{align*}
Here, \(\mathcal{B} := E_{ik} V\).  The inputs \(u_1(t)\) and \(u_2(t)\) are composed of sinusoids at the frequencies \((\pm \Omega_l \pm \omega)\), which may destabilize the system if (following conditions~\eqref{eq:basic-instability-criterion-filtered-plus} and \eqref{eq:basic-instability-criterion-filtered-minus}) this equals \(\pm {\scriptstyle \Omega}_m\), one of the modes of the unforced system.  However, this also requires the system to be controllable from those inputs.  For instance, when \(j {\scriptstyle \Omega}_m = j \left( {\scriptstyle \Omega}_l + \omega \right)\) for some \(l, m\), we have
\begin{align}
\ddot{\phi}_1 = & -L \phi_1(t) - \tfrac{1}{2} \underbrace{E_{ik} \mathbf{v}_l}_{\mathcal{B}_l} \underbrace{e^{j (\Omega_l + \omega) t} c_l}_{u_l(t)} \nonumber\\
& + \text{other inputs} \label{eq:controllability-second-order-illustration}
\end{align}
where the other inputs include the complex conjugate of \(u_l(t)\), as well as inputs that do not destabilize the system.  Applying the eigenvector controllability test to system~\eqref{eq:controllability-second-order-illustration} rewritten in first order form (\cref{sec:eigenvector-test-for-controllability}) gives us the following criterion for uncontrollability:
\begin{align}
  \begin{bmatrix} \mathbf{v}_{m}^{\star}  & \text{-} j {\scriptstyle \Omega}_m^{\text{-}1} \mathbf{v}_{m}^{\star} \end{bmatrix} \begin{bmatrix} 0 & 0 \\ E_{ik} & 0 \end{bmatrix} \begin{bmatrix} \mathbf{v}_{l} \\ j {\scriptstyle \Omega}_l \mathbf{v}_{l} \end{bmatrix} & \stackrel{?}{=} 0 \nonumber\\ \implies \mathbf{v}_m^{\star} E_{ik} \mathbf{v}_l & \stackrel{?}{=} 0. \label{eq:controllability-analysis-criterion}
\end{align}
When this is the case, the mode with frequency \(\pm j {\scriptstyle \Omega}_m\) is not controllable from \(u_l(t)\), and thus the harmonic resonance predicted by the instability criteria~\eqref{eq:basic-instability-criterion-filtered-plus} and \eqref{eq:basic-instability-criterion-filtered-minus} cannot occur.  The controllability-based criterion~\eqref{eq:controllability-analysis-criterion} thus provides an alternative interpretation of the criterion obtained by the multiple-scale perturbation analysis \eqref{eq:multiscale-analysis-full-criteria}.

\subsection{Generalization to Subnetwork Excitation}
\label{sec:generalization-partial-network-forcing}

The methods described in sections~\ref{sec:full-network-parametric-forcing}, \ref{sec:regular-perturbation-analysis} and~\cref{sec:vector-valued-multiple-scale-perturbation-analysis} can be extended to the parametric forcing of any subnetwork \(P\) of the graph \(L\), where \(P\) and \(L\) are as in \cref{eq:linearized-swing-dynamics-perturbed}.  The critical frequencies are still subject to the same conditions~\eqref{eq:basic-instability-criterion-filtered-Oeps2-plus} and \eqref{eq:basic-instability-criterion-filtered-Oeps2-minus}, so they are unchanged.  The perturbation analysis is unchanged but for replacing the single edge Laplacian \(E_{ik}\) with \(P\) in \eqref{eq:multiscale-analysis-full-criteria}.  The slopes of the first-order Arnold tongues are thus given by
\begin{align}
\label{eq:partial_network_forcing_slope_criterion}
\abs{a_{lm}} = & \frac{\abs{\mathbf{v}_m^{\star} P \mathbf{v}_l}}{2 \left( {\scriptstyle \Omega}_l {\scriptstyle \Omega}_m \right)^{1/2}}
\end{align}
When the subnetwork \(P\) in \eqref{eq:linearized-swing-dynamics-perturbed} is the full network \(L\), we expect the perturbation analysis criterion~\eqref{eq:multiscale-analysis-full-criteria} to reduce to the results of \cref{sec:full-network-parametric-forcing}.  When \(P = L\), all cross terms of the form \(\mathbf{v}_m^{\star} P \mathbf{v}_l = \mathbf{v}_m^{\star} L \mathbf{v}_l\) with \(m \ne l\) vanish, since \(\mathbf{v}_m^{\star} L \mathbf{v}_l = {\scriptstyle \Omega}_l^2\  \mathbf{v}_m^{\star} \mathbf{v}_l = 0\).  Thus the only critical frequencies are of the form \(\omega_0 = \abs{{\scriptstyle \Omega}_l + {\scriptstyle \Omega}_l}/n = 2 {\scriptstyle \Omega}_l / n\), for \(n = 1, 2, \dots, {\scriptstyle N}\).  These are the parametric resonance frequencies of \({\scriptstyle N}\) decoupled Mathieu oscillators, which agrees with the diagonalization argument of \cref{sec:full-network-parametric-forcing}.

\section{Stability Analysis for  Tori Graphs}
\label{sec:ring-network-susceptibility-analysis}

As an example, we apply the stability analysis results of \cref{sec:single-edge-parametric-forcing} to ring and tori graphs.  We can obtain analytical expressions or bounds for all the frequencies and tongue widths for parametric resonance on such networks, which gives us some insight into how the stability characteristics vary with increasing graph size and connectivity.

We begin with the one-dimensional torus, i.e. a ring graph. 
With \(\theta := \frac{2 \pi}{N}\), the eigenvalues and eigenvectors of the \({\scriptstyle N}\)-node ring graph are~\cite{bamieh2012Coherence}
\begin{align}
{\scriptstyle \Omega}_m^2 & = 2 ( 1 - \cos(m \theta)), \quad m = 0, \dots, {\scriptstyle N}\text{-}1 \label{eq:ring-graph-eigenvalues}\\
\mathbf{v}_m & = \begin{bmatrix} 1 & e^{j m \theta} & e^{j 2 m \theta} & \dots & e^{j (N\text{-}1) m \theta} \end{bmatrix}^\mathrm{T} \label{eq:ring-graph-eigenvectors}
\end{align}
For the ring graph, we can thus find the critical frequencies~\eqref{eq:basic-instability-criterion-filtered-plus} and \eqref{eq:basic-instability-criterion-filtered-minus} and the slopes of the corresponding Arnold tongues~\eqref{eq:multiscale-analysis-full-criteria} analytically.  The frequency corresponding to the eigenvalue pair \((l,m)\) is
\begin{align}
\omega_{lm} = & {\scriptstyle \Omega}_l + {\scriptstyle \Omega}_m, \quad  (l, m = 1, \dots, {\scriptstyle N} \text{-}1)\nonumber\\
= & 2^{1/2} \left( ( 1 - \cos(l \theta) )^{1/2} + (1 - \cos(m \theta))^{1/2} \right) \label{eq:ring-graph-critical-frequencies}
\end{align}
Since \({\scriptstyle \Omega}_m^2 = {\scriptstyle \Omega}_{N \text{-} m}^2\), all eigenvalues of the Laplacian (except possibly one) have multiplicity \(2\).  The corresponding eigenvectors are the complex conjugates \(\mathbf{v}_m\) and \(\mathbf{v}_{N \text{-} m} = \mathbf{\bar{v}}_m\).

The first-order Arnold tongues in \( \omega, \epsilon \) space are then given by
\begin{align}
\omega(\epsilon) = \omega_{lm} \pm a_{lm} \epsilon + \mathcal{O}(\epsilon^2) \label{eq:ring-graph-arnold-tongue-equation}
\end{align}
where \( \pm a_{lm} \) is the slope of the tongue at critical frequency \( \omega_{lm} \).  \( a_{lm} \) is the maximum attainable value of \({\mathbf{v}_m^{\star} E_{ik} \mathbf{v}_l}/{2 \left( {\scriptstyle \Omega}_l {\scriptstyle \Omega}_m \right)^{1/2}}\) (\cref{sec:multi-scale-method-with-multiplicity}), where \(E_{ik} = e_{ik} e_{ik}^{\star}\) is the single edge (\(i \to k\)) Laplacian corresponding to the parametrically forced edge, and \(\mathbf{v}_m\) and \(\mathbf{v}_l\) range over the corresponding \(2\)-dimensional eigenspaces.  If \(m \ne {\scriptstyle N}/2\) and \(l \ne {\scriptstyle N}/2\),
\begin{align}
\abs{a_{lm}} = & \frac{\left\| \begin{pmatrix} \mathbf{v}_m^{\star} \\ \mathbf{v}_{N \text{-} m}^{\star} \end{pmatrix} E_{ik} \begin{pmatrix} \mathbf{v}_l & \mathbf{v}_{N \text{-} l} \end{pmatrix} \right\|_2 }{2 \left( {\scriptstyle \Omega}_l {\scriptstyle \Omega}_m \right)^{1/2}} \nonumber\\
= & \frac{\left\| \begin{matrix} \mathbf{v}_m^{\star} e_{ik} \\ \mathbf{\bar{v}}_m^{\star} e_{ik} \end{matrix} \right\|_2 \left\| \begin{matrix} e_{ik}^{\star} \mathbf{v}_l & e_{ik}^{\star} \mathbf{\bar{v}}_l \end{matrix} \right\|_2}{2 \left( {\scriptstyle \Omega}_l {\scriptstyle \Omega}_m \right)^{1/2}} \nonumber\\
  = & \frac{\abs{\mathbf{v}_m^{\star} e_{ik}}}{{\scriptstyle \Omega}_l^{1/2}} \frac{\abs{e_{ik}^{\star} \mathbf{v}_l}}{{\scriptstyle \Omega}_m^{1/2}}
= \frac{\abs{\mathbf{v}_m^{\star} E_{ik} \mathbf{v}_l }}{\left( {\scriptstyle \Omega}_l {\scriptstyle \Omega}_m \right)^{1/2}} \label{eq:ring-graph-arnold-tongue-slopes-bound-generic} \\
  = &  \frac{1}{{\scriptstyle N}} \frac{\abs{1 - e^{\text{-}j m \theta}}}{ (2 \text{-} 2 \cos(\text{-}m \theta)) ^{1/4}} \frac{\abs{1 - e^{j l \theta}}}{(2 \text{-} 2 \cos(l \theta))^{1/4}} \label{eq:ring-graph-arnold-tongue-slopes-bound}
\end{align}
If \(m = {\scriptstyle N}/2\) or \(l = {\scriptstyle N}/2\), then the corresponding eigenspace is one-dimensional, and the slope \(\abs{a_{lm}}\) is lower by a factor of \(2^{1/2}\):
\begin{align}
  \abs{a_{lm}} & = \frac{\abs{\mathbf{v}_m^{\star} E_{ik} \mathbf{v}_l}}{\nu \left( {\scriptstyle \Omega}_l {\scriptstyle \Omega}_m \right)^{1/2}}, \text{ with} \label{eq:ring-graph-arnold-tongue-slopes-adjusted-bound}\\
    \nu & = \begin{cases}
      1 & l \ne {\scriptstyle N}/2, m \ne {\scriptstyle N}/2 \\
      2^{1/2} & l = {\scriptstyle N}/2 \text{ or } m = {\scriptstyle N}/2 \\
      2 & l = m = {\scriptstyle N}/2
    \end{cases} \nonumber
\end{align}
Expression~\eqref{eq:ring-graph-arnold-tongue-slopes-bound} is separable and symmetric in \(m\) and \(l\), so we can find the most critical frequencies by maximizing the function \(\delta(t) := \abs{1 - e^{j t}}(2 - 2 \cos(t))^{\text{-}1/4}\), where \(\abs{a_{lm}} = \frac{1}{\nu {\scriptstyle N}} \delta(\text{-}m \theta) \delta(l \theta)\).  This function attains a maximum of \(2^{1/2}\) at \(t = \pi\), corresponding to \(l = m = \lfloor {\scriptstyle N}/2 \rfloor\).  We make the following observations about the stability of the ring graph under parametric excitation of one of its edges:
\begin{itemize}
\item All critical frequencies for parametric instability at small forcing amplitudes are contained in \(\omega \in (0, 4]\) (\cref{eq:ring-graph-eigenvalues}).  The critical frequencies are independent of the choice of parametrically forced edge(s).
\end{itemize}

\begin{itemize}
\item The size of the unstable region in \((\omega, \epsilon)\) space around any critical frequency is inversely proportional to the number of nodes \({\scriptstyle N}\) in the graph.
\item As measured from the vertical line, the widest Arnold tongue in the parameter space \((\omega, \epsilon)\) has slope between \({\scriptstyle 1/N}\) and \(2/{\scriptstyle N}\).
\end{itemize}

\cref{fig:ring-graph-stability-diagrams-analytical} shows the \(\mathcal{O}(\epsilon)\) approximation to the Arnold tongues (\cref{eq:ring-graph-arnold-tongue-equation}) for ring networks of varying sizes.

\begin{figure}[t]
\centering
\includegraphics[width=0.9\columnwidth]{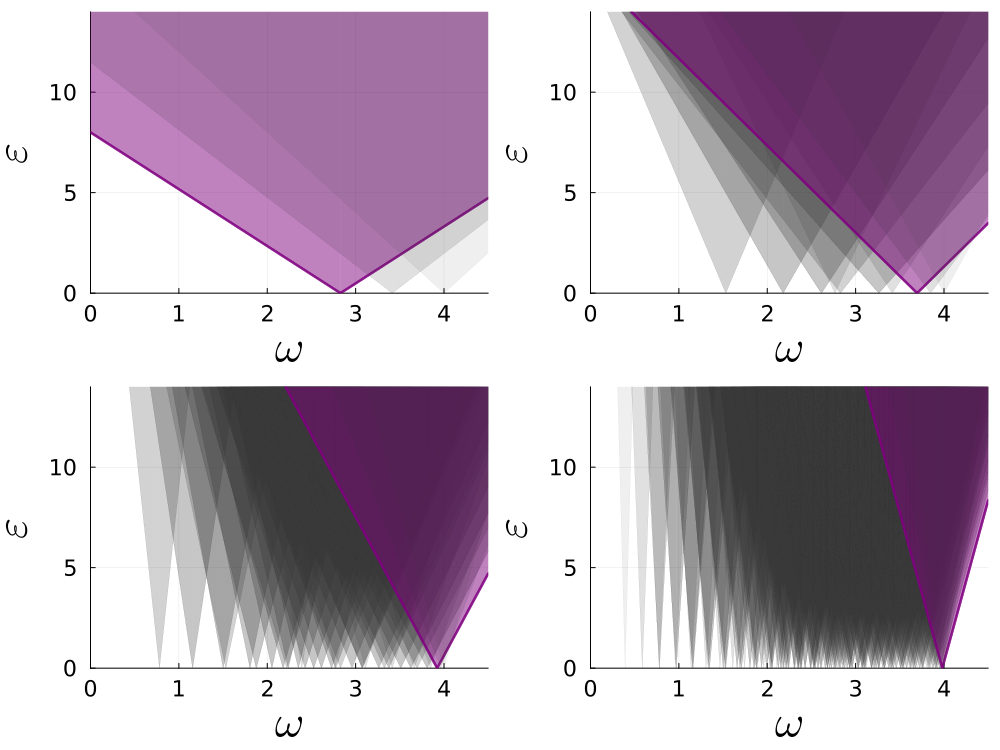}
\caption{\label{fig:ring-graph-stability-diagrams-analytical}Visualization of the stability diagram for ring networks with single edge perturbation.  The four networks considered have (top-left) \({\scriptstyle N} = 4\), (top-right) \(8\), (bottom-left) \(16\) and \(24\) nodes.  All parametric forcing frequencies causing parametric resonance are contained in the interval \((0, 4]\).  The most critical frequency corresponds to the widest Arnold tongue and is colored purple.  For sufficiently large \({\scriptstyle N}\), this tongue is at \(\omega = 4\) or close to it, and the slopes of all tongues as measured from the vertical are inversely proportional to \({\scriptstyle N}\), the number of nodes.}
\end{figure}

\subsection{Ring Networks with more General Forcing}

These analytical results for ring networks can be extended to the case of more general forcing, when more than one edge is forced parametrically at a single frequency \(\omega\).  Let \(P\) be the graph Laplacian corresponding to the parametrically forced edges, so that the full graph Laplacian is \(L + \epsilon P \cos(\omega t)\) (Section~\ref{sec:problem-statement}).  With some abuse of notation, we refer to the graph corresponding to \(P\) as \(P\) as well.  The perturbation analysis is unchanged but for replacing the single edge Laplacian \(E_{ik}\) with \(P\).  As before, with \({\scriptstyle \Omega}_m^2 = {\scriptstyle \Omega}_l^2\) and \(\mathbf{v}_m = \mathbf{\bar{v}}_{N \text{-} m}\)the slopes of the Arnold tongues are given by (Section~\ref{sec:generalization-partial-network-forcing})
\begin{align}
  \abs{a_{lm}} = & \frac{\left\|\begin{pmatrix} \mathbf{v}_m^{\star} \\ \mathbf{v}_{{\scriptscriptstyle N} \text{-}m}^{\star} \end{pmatrix} P \begin{pmatrix} \mathbf{v}_l & \mathbf{v}_{N \text{-} l} \end{pmatrix} \right\|_2}{2 \left( {\scriptstyle \Omega}_l {\scriptstyle \Omega}_m \right)^{1/2}}  \nonumber\\
  = & \frac{1}{2 \left( {\scriptstyle \Omega}_l {\scriptstyle \Omega}_m \right)^{1/2}} \left\| \begin{matrix} \mathbf{v}_m^{\star} P \mathbf{v}_l & \mathbf{v}_m^{\star} P \mathbf{\bar{v}}_l \nonumber\\
  \mathbf{\bar{v}}_m^{\star} P \mathbf{v}_l & \mathbf{\bar{v}}_m^{\star} P \mathbf{\bar{v}}_l \end{matrix} \right\|_2 \\
  = &  \frac{2^{1/2}}{2 \left( {\scriptstyle \Omega}_l {\scriptstyle \Omega}_m \right)^{1/2}} \left\| \begin{matrix}\mathbf{v}_m^{\star} P \mathbf{v}_l \\ \mathbf{\bar{v}}_m^{\star} P \mathbf{v}_l\end{matrix} \right\|_2 \nonumber\\
= & \frac{\left( \abs{\mathbf{v}_m^{\star} P \mathbf{v}_l}^2 + \abs{\mathbf{\bar{v}}_m^{\star} P \mathbf{v}_l}^2 \right)^{1/2}}{\left( 2 {\scriptstyle \Omega}_l {\scriptstyle \Omega}_m \right)^{1/2}}
\end{align}
Accounting for one-dimensional eigenspaces corresponding to \(m = {\scriptstyle N}/2\) or \(l = {\scriptstyle N}/2\), we get
\begin{align}
  \abs{a_{lm}} & = \frac{\left( \abs{\mathbf{v}_m^{\star} P \mathbf{v}_l}^2 + \abs{\mathbf{\bar{v}}_m^{\star} P \mathbf{v}_l}^2 \right)^{1/2}}{\nu \left( 2 {\scriptstyle \Omega}_l {\scriptstyle \Omega}_m \right)^{1/2}} \label{eq:ring-graph-general-forcing-arnold-tongues-adjusted-bound}
\end{align}
where \(\nu\) is the same as in \cref{eq:ring-graph-arnold-tongue-slopes-adjusted-bound}.

Since \(P\) can be written as the sum of single edge graph Laplacians (\(P = \sum\limits_{i \to k} E_{ik}\), where \(i \to k\) are the edges of \(P\) with \(\abs{i - k} = 1\)), we can simplify each term in this expression further:
\begin{align*}
\abs{\mathbf{v}_m^{\star} P \mathbf{v}_l} = & \abs{\sum\limits_{
\substack{
i \to k\\
\abs{i \text{-} k} = 1}
} \mathbf{v}_m^{\star} E_{ik} \mathbf{v}_l} \\
= & \frac{1}{{\scriptstyle N}} \abs{ \left( 1 - e^{j \theta l} \right) \left( 1 - e^{\text{-}j \theta m} \right) \sum\limits_{i \to k} e^{j \theta i (l \text{-} m)}},
\end{align*}
with a similar expression for \(\abs{\mathbf{\bar{v}}_m^{\star} P \mathbf{v}_l}\).  Both expressions attain a maximum of \(4 {\scriptstyle M/N}\) at \(l = m = \left\lfloor N/2 \right\rfloor\), allowing us to draw the following conclusions.
\begin{itemize}
\item As in the single edge case, all critical frequencies for parametric instability at small forcing amplitudes are contained in \(\omega \in (0, 4]\).
\item The size of the unstable region in \((\omega, \epsilon)\) space near any critical frequency is inversely proportional to the number of nodes \({\scriptstyle N}\), and the the widest tongue has slope (as measured from the vertical) between \({\scriptstyle M/N}\) and \(2 {\scriptstyle M/N}\).
\end{itemize}

When the whole network is perturbed (\({\scriptstyle M = N}\)), the analysis of \cref{sec:full-network-parametric-forcing} applies, and the slope of the widest Arnold tongue is \(\pm 1\), a well known result for Mathieu's equation.

\subsection{Higher-Dimensional Tori}

We can study the effect of the connectivity of the ring network on its stability under parametric edge excitation by extending the results \eqref{eq:ring-graph-arnold-tongue-slopes-bound} and  \eqref{eq:ring-graph-general-forcing-arnold-tongues-adjusted-bound} to tori in higher dimensions.  We consider a ring network that is a \(d\)-fold periodic lattice with \({\scriptstyle N}^{d}\) nodes, where each node has degree \(d\).  The Laplacian of this graph is an \({\scriptstyle N}^{d} \times {\scriptstyle N}^{d}\) matrix, whose spectral properties govern the slopes of the Arnold tongues as in \eqref{eq:multiscale-analysis-full-criteria}.  Since this matrix is \(d\)-fold circulant~\cite{bamieh2012Coherence}, its action on a vector \(\mathbf{v}\) is equivalent to a \(d\)-fold circular convolution:
\begin{align*}
L_a \mathbf{v} & = \mathbf{a} \star \mathbf{v} \\
\left( L_a \mathbf{v} \right)_{(p_1, \dots, p_d)} & = \sum\limits_{\substack{q_1, \dots, q_d \in \mathbb{Z}_N}} \mathbf{a}_{(p_1 - q_1), \dots, (p_d - q_d)} \mathbf{v}_{(q_1, \dots, q_d)} \\
\implies \left( L_a \mathbf{v} \right)_{\mathbf{p}} & = \sum\limits_{\mathbf{q} \in \mathbb{Z}_N^d} \mathbf{a}_{\mathbf{p} - \mathbf{q}} \mathbf{v}_{\mathbf{q}}
\end{align*}
where \(L_a\) is the Laplacian and \(\mathbf{a}\) is one of its columns that fully determines \(L_a\).  We index \(\mathbf{a}\) and \(\mathbf{v}\) as multidimensional arrays with \(d\) indices \(\mathbf{p} := \left( p_1, \dots, p_d \right)\) and \(\mathbf{q} := \left( q_1, \dots, q_n \right)\) in \(\mathbb{Z}_N^d\), which makes the convolution (and thus the spectrum of \(L_a\)) convenient to represent.  The eigenvectors of \(L_a\) are the eigenvectors of the \(d\)-fold circular shift operator, and the eigenvalues are the entries of the \(d\)-dimensional DFT of \(\mathbf{a}\).  Indexing the eigenvalues by \((m_1, \dots, m_d) =: \mathbf{m} \in \mathbb{Z}_N^d\), we have
\begin{align}
{\scriptstyle \Omega}^2_{\mathbf{m}}  & = 2 \sum\limits_{i=1}^d \left( 1 \text{-} \cos(m_i \theta) \right)) \label{eq:ring-graph-d-fold-eigenvalues} \\
\mathbf{v}^{(\mathbf{m})}_{(q_1, q_2, \dots, q_d)} & = e^{j \theta \left( m_1 q_1 + \dots + m_d q_d \right)} \nonumber\\
\implies \mathbf{v}^{(\mathbf{m})}_{\mathbf{q}} & = e^{j \theta \mathbf{m}\cdot\mathbf{q}} \label{eq:ring-graph-d-fold-eigenvectors}
\end{align}
where \(\mathbf{v}^{(\mathbf{m})}\) is the eigenvector corresponding to \({\scriptstyle \Omega}_{\mathbf{m}}^2\) and \(\theta = 2 \pi / {\scriptstyle N}\), as before.  The critical frequencies \eqref{eq:basic-instability-criterion-filtered-plus} and \eqref{eq:basic-instability-criterion-filtered-minus} are then
\begin{align}
& \omega_{\mathbf{l} \mathbf{m}} = {\scriptstyle \Omega}_{\mathbf{l}} + {\scriptstyle \Omega}_{\mathbf{m}} \nonumber\\
& = \sqrt{2}  \left[ \left( \sum\limits_{i=1}^d 1 \text{-} \cos(l_i \theta) \right)^{\frac{1}{2}} + \left( \sum\limits_{i=1}^d 1 \text{-} \cos(m_i \theta) \right)^{\frac{1}{2}} \right] \label{eq:ring-graph-d-fold-critical-frequencies}
\end{align}
Because of the \(d\)-fold symmetry in the expression for \({\scriptstyle \Omega}_{\mathbf{m}}^2\) (\(\cos(m_i \theta) = \cos(({\scriptstyle N}\text{-} m_i)\theta)\)), the eigenvalue \({\scriptstyle \Omega}_{\mathbf{m}}^2\) has multiplicity \(2^d\), and the slope of the corresponding Arnold tongue involves the \(2\)-norm of a matrix (\cref{sec:multi-scale-method-with-multiplicity}).  Let \(V_{\mathbf{m}}\) and \(V_{\mathbf{l}}\) be matrices whose columns span the eigenspace corresponding to \({\scriptstyle \Omega}_{\mathbf{m}}^2\) and \({\scriptstyle \Omega}_{\mathbf{l}}^2\) respectively, and \(E_{\mathbf{ik}} \equiv e_{\mathbf{ik}} e_{\mathbf{ik}}^{\star}\) be the single edge Laplacian corresponding to the edge \(\mathbf{i} \leftrightarrow \mathbf{k}\), where \(\mathbf{i}, \mathbf{k} \in \mathbb{Z}_N^d\).  Note that by the topology of the ring lattice, \(\mathbf{i}\) and \(\mathbf{k}\) are identical in all indices except for one, which differs by one.  Without loss of generality, we may assume this is the first index.  The slopes of the Arnold tongues (\cref{eq:multiscale-analysis-full-criteria}) corresponding to perturbations in edge \(\mathbf{i} \to \mathbf{k}\) are then
\begin{align}
\abs{a_{\mathbf{l}\mathbf{m}}} = & \max_{\abs{\mathbf{z}^{(\mathbf{m})}} = \abs{\mathbf{z}^{(\mathbf{l})}} = 1} \frac{\mathbf{z}^{(\mathbf{m}) \star} V_{\mathbf{m}}^{\star} E_{ik} V_{\mathbf{l}} \mathbf{z}^{(\mathbf{l})}}{2 \left( {\scriptstyle \Omega}_\mathbf{l} {\scriptstyle \Omega}_\mathbf{m} \right)^{1/2}}\nonumber\\
  = & \frac{ \|V_{\mathbf{m}}^{\star} E_{ik} V_{\mathbf{l}} \|_2 }{2 \left( {\scriptstyle \Omega}_\mathbf{l} {\scriptstyle \Omega}_\mathbf{m} \right)^{1/2}} \label{eq:ring-graph-d-fold-arnold-tongue-slopes}\\
  \ge & \frac{\abs{\mathbf{v}^{(\mathbf{m})\star} E_{\mathbf{ik}} \mathbf{v}^{(\mathbf{l})}}}{2 \left( {\scriptstyle \Omega}_\mathbf{l} {\scriptstyle \Omega}_\mathbf{m} \right)^{1/2}} = \frac{\abs{(\mathbf{v}^{(\mathbf{m})\star} e_{\mathbf{ik}})(e_{\mathbf{ik}}^{\star} \mathbf{v}^{(\mathbf{l})})}}{2 \left( {\scriptstyle \Omega}_\mathbf{l} {\scriptstyle \Omega}_\mathbf{m} \right)^{1/2}} \label{eq:ring-graph-d-fold-arnold-tongue-slopes-lower-bound-generic} \\
= & \frac{1}{2 {\scriptstyle N}^d} \frac{\abs{e^{j \theta \mathbf{m} \cdot \mathbf{i}}}\abs{1 - e^{\text{-}j m_1 \theta}}}{ {\scriptstyle \Omega}_{\mathbf{m}}^{1/2}} \frac{\abs{e^{j \theta \mathbf{l} \cdot \mathbf{i}}}\abs{1 - e^{j l_1 \theta}}}{{\scriptstyle \Omega}_{\mathbf{l}}^{1/2}} \nonumber\\
= & \frac{1}{2 {\scriptstyle N}^d} \frac{\abs{1 - e^{\text{-}j m_1 \theta}}}{ {\scriptstyle \Omega}_{\mathbf{m}}^{1/2}} \frac{\abs{1 - e^{j l_1 \theta}}}{{\scriptstyle \Omega}_{\mathbf{l}}^{1/2}} \label{eq:ring-graph-d-fold-arnold-tongue-slopes-lower-bound}
\end{align}
This is analogous to expression \eqref{eq:ring-graph-arnold-tongue-slopes-bound} for the one-dimensional ring, except in that we derive a lower bound on the slope of the widest Arnold tongue instead of an exact estimate.  This lower bound is inversely proportional to the total number of nodes (\({\scriptstyle N}^{d}\) here), and attains its maximum of \(1/{\scriptstyle N}^d\) when \(m_1 = l_1 = \left\lfloor {\scriptstyle N}/2 \right\rfloor\), \(m_q = l_q = 0,\ q = 2, \dots, d\).  With the number of nodes fixed, this behavior is identical to the one-dimensional case and independent of the dimension of the lattice.  All critical frequencies are contained in \(\omega \in (0, 4 d^{1/2} ]\), and the widest Arnold tongue in the parameter space \((\omega, \epsilon)\) thus has slope larger than \(1/{\scriptstyle N}^d\).

\cref{fig:ring-graph-stability-diagrams-analytical-2d-comparison} shows the \(\mathcal{O}(\epsilon)\) approximation to the Arnold tongues for ring lattices of various sizes with \(d = 2\).  As a result of the eigenvalue multiplicity, there are fewer critical frequencies, and the corresponding Arnold tongues are wider than for the one-dimensional ring graph with the same number of nodes.

\begin{figure}[t]
\centering
\includegraphics[width=0.9\columnwidth]{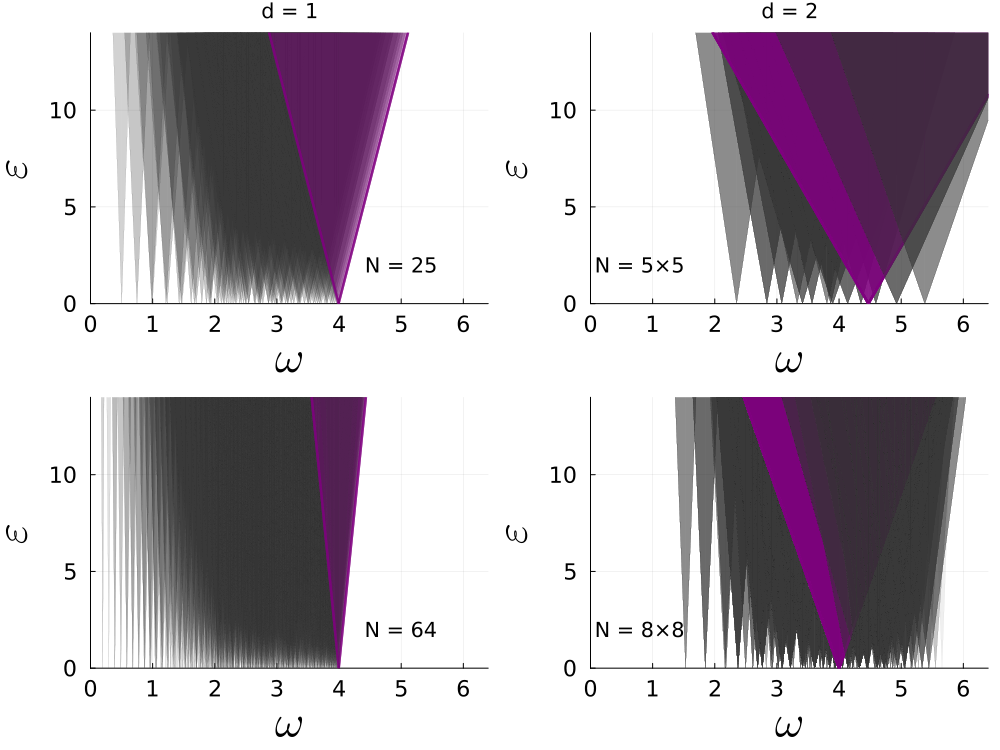}
\caption{\label{fig:ring-graph-stability-diagrams-analytical-2d-comparison}Comparison of the stability diagram for ring lattices of dimension \(1\) (left) and \(2\) (right).  The top row is for graphs with 25 nodes, the bottom row for graphs with 64.  The widest Arnold tongue in each case is colored purple.  The Laplacian for a two-dimensional lattice has much higher eigenvalue multiplicity, the effect of which is to produce fewer critical frequencies and wider Arnold tongues.  (Compare with \cref{fig:ring-graph-stability-diagrams-analytical-2d-comparison-asymmetric}.)}
\end{figure}
\subsection{The Effect of Connectivity}

Narrower Arnold tongues are harder to observe in physical systems, especially in the presence of damping.  Thus \cref{eq:ring-graph-d-fold-arnold-tongue-slopes} suggests that the increasing the connectivity of the graph has a destabilizing effect, as the tongues are wider (\cref{fig:ring-graph-stability-diagrams-analytical-2d-comparison}).  However, this is a consequence of the effect of the eigenvalue multiplicity, a non-generic factor, and not the increased connectivity of the graph.  We investigate this by perturbing the edge weights of the ring graph slightly, which makes the eigenvalues distinct in the generic case.  This turns the lower bound calculations for the slope of the widest tongues (\cref{eq:ring-graph-arnold-tongue-slopes-bound-generic} and \cref{eq:ring-graph-d-fold-arnold-tongue-slopes-lower-bound-generic}) into equalities, and ring graph connectivity has little to no effect on the distribution of tongue slopes.  \cref{fig:ring-graph-stability-diagrams-analytical-2d-comparison-asymmetric} compares the stability diagrams for one and two-dimensional ring graphs with the same number of nodes but with the effect of eigenvalue multiplicity removed.

\begin{figure}[t]
\centering
\includegraphics[width=0.9\columnwidth]{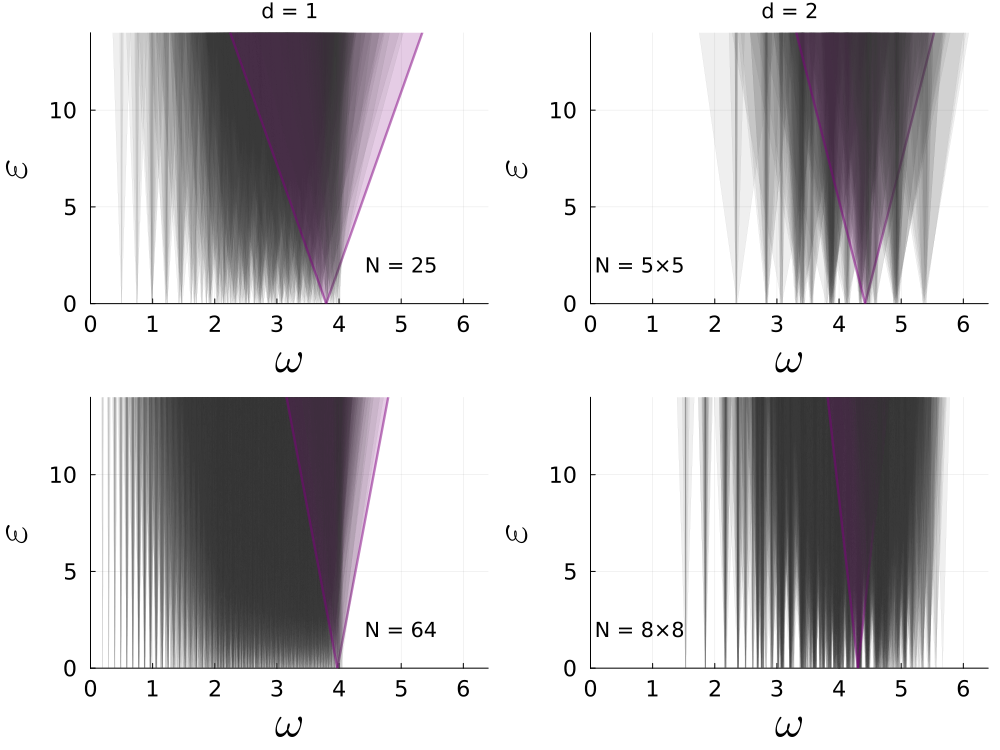}
\caption{\label{fig:ring-graph-stability-diagrams-analytical-2d-comparison-asymmetric}Comparison of the stability diagram (to \(\mathcal{O}(\epsilon)\)) for ring lattices of dimension \(1\) (left) and \(2\) (right) without eigenvalue multiplicity.  The top row is for graphs with 25 nodes, the bottom row for graphs with 64, and the nominal edge weights of \(1\) are perturbed by a normal variable with standard deviation \(0.01\), which makes the Laplacian eigenspaces one-dimensional in the generic case.  The widest Arnold tongue in each case is colored purple.  With the effects of eigenvalue multiplicity removed, the slopes of the widest tongues for the two cases are very close to each other, implying a negligible effect of graph connectivity on the stability of the dynamics.  The widest tongues are no longer guaranteed to be near \(\omega = 4\), since the graph Laplacians are not circulant.  Compare with \cref{fig:ring-graph-stability-diagrams-analytical-2d-comparison}.}
\end{figure}

\cref{ring-graph-asymmetric-1d-vs-2d-all-slopes} compares all the slopes and the number of first-order Arnold tongues for one and two-dimensional ring graphs with perturbed edge weights and various graph sizes.  The increased connectivity of the \(2\)-dimensional lattice does not affect the stability of the dynamics.

\begin{figure}[t]
\centering
\includegraphics[width=0.9\columnwidth]{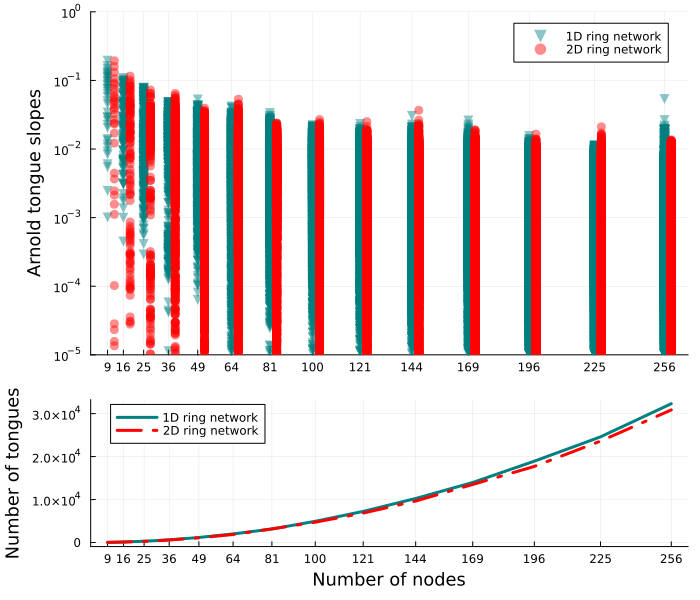}
\caption{\label{ring-graph-asymmetric-1d-vs-2d-all-slopes}(Above) Slopes of all Arnold tongues for one and two-dimensional ring graphs of the same size with perturbed edge weights, as a function of the size of the network.  The two are offset slightly for legibility.  (Below) The number of Arnold tongues for both networks as a function of the network size.  There is no significant difference in the distribution of slopes or the number of distinct Arnold tongues between the two kinds of graphs.}
\end{figure}
\section{Conclusion and Further Questions}
\label{sec:concl-future-work}

We draw the following tentative conclusions about parametric resonance via periodic edge perturbations in networked oscillators:
\begin{itemize}
\item The stability characteristics of parametric resonance via edge perturbations of linear oscillator dynamics on undirected graphs can be inferred from the spectrum of the graph Laplacian.  The critical parametric forcing frequencies are sums of pairs of the natural frequencies of the graph, which are themselves the square roots of the Laplacian eigenvalues.
\item Each critical frequency is thus associated with two eigenvalues.  The slope or width of the Arnold tongue structure in parameter space originating at that frequency is a function of the corresponding eigenvectors.  It is large when a corresponding eigenvector varies significantly between the nodes connecting the perturbed edges, or when the eigenvalues have high multiplicity.  If either eigenvector has the same value on these nodes, there is no Arnold tongue at the corresponding frequency and the dynamics remains stable.
\item The Arnold tongue slopes are generally proportional to the number of edges being perturbed, and inversely proportional to the size of the network.  They do not appear to depend strongly on the degree of network connectivity.
\end{itemize}

While we choose parametric forcing at a single frequency to illustrate the method, these results extend to more general periodic forcing.  The perturbation analysis treats the forcing as additive input to a linear system, so each harmonic of the forcing period acts as a separate input.  The system thus experiences parametric resonance if any harmonic of the forcing frequency equals the sum of two oscillator natural frequencies.

Parametric resonance in real-world oscillator networks differs from the models considered in this work because of damping.  In parametric oscillators the effect of damping on the stability diagrams is to ``raise'' the Arnold tongues (\cref{fig:mathieu-instability-rescaled}), an effect more pronounced in narrower or higher-order tongues.  Thus only the resonances corresponding to the wider and lower-order tongues are observable in physical systems.  Approximating these realizable tongue positions and widths for the purposes of design, to avoid or excite parametric resonance, can be a computationally intensive undertaking for large networks.  The multiple time-scale perturbation analysis provides a simple criterion for the same that depends only on knowledge of the graph Laplacian.

A natural extension of the analysis presented here is thus to consider the effect of damping.  In its current form, the perturbation analysis depends on a diagonalization argument and can be readily extended to cover only very specific cases of damping: the damping coefficient matrix must commute with the graph Laplacian.  This is the case, for example, when each oscillator experiences ``self-damping'' with the same damping coefficient.  A new approach is required for a more general treatment.

It would be useful to know if the multiple scale perturbation method can predict the unstable modes in addition to the positions and slopes of the Arnold tongues.  Preliminary analysis suggests that this is the case.  To first order in the forcing strength, the position and slope of each Arnold tongue depends on two modes (eigenvalues and eigenvectors) of the graph Laplacian, and the unstable slow-time dynamics in these tongues is a combination of these modes as well (see \cref{sec:multi-scale-method-unstable-modes}).  \cref{fig:unstable-response-comparison-consolidated} compares, for a ring graph, the actual unstable mode obtained via numerical computation of the Monodromy map with the corresponding modes of the Laplacian used in the perturbation analysis.  We observe a qualitative similarity between the actual and predicted mode shapes.

\begin{figure}[t]
\centering
\includegraphics[width=0.9\columnwidth]{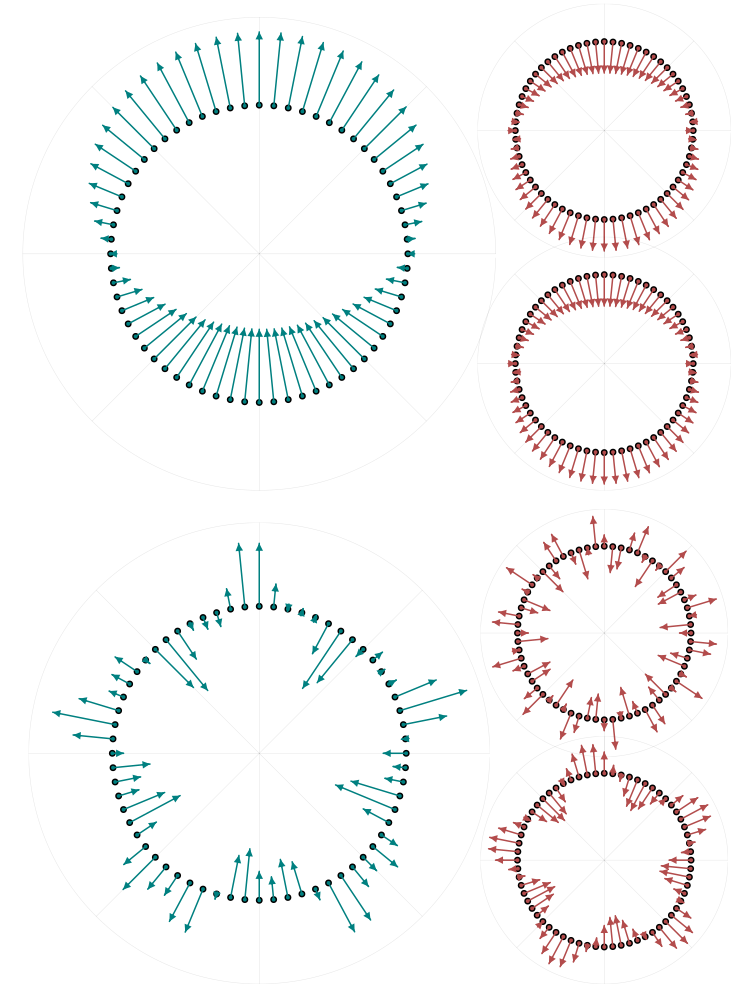}
\caption{\label{fig:unstable-response-comparison-consolidated}Comparison of the unstable mode shapes as predicted by the numerically computed Monodromy map (left, large circle) and the multiscale perturbation analysis (right, small circles) for a \(64\) node ring graph with one time-periodically varying edge weight, at two critical parametric forcing frequencies (\cref{eq:multiscale-analysis-full-criteria}).  The effect of eigenvalue multiplicity has been removed by perturbing the nominal edge weights slightly, and the value of the mode at each node is plotted radially.  Above and below: this comparison is shown at two critical parametric forcing frequencies.  At both frequencies, the mode shape as predicted by the perturbation analysis appears to be in the span of the two eigenvectors of the Laplacian corresponding to the critical frequency (\cref{sec:multi-scale-method-unstable-modes}).  We observe that the true unstable modes are visually similar to one or a combination of the two Laplacian eigenvectors.}
\end{figure}

\bibliographystyle{IEEEtran}
\bibliography{./biblio}

\appendix
\section{Appendix}
\subsection{Vector-valued multiple scale perturbation analysis}
\label{sec:vector-valued-multiple-scale-perturbation-analysis}

When one of the conditions \eqref{eq:basic-instability-criterion-filtered-plus} or \eqref{eq:basic-instability-criterion-filtered-minus} is met, we investigate the behavior of the system near a critical \(\omega\).

To facilitate the multiple time-scale analysis, we rescale time so that the parametric forcing frequency \(\omega\) is uniformly \(1\).  Under the transformation \(\eta \leftarrow \omega t\), we get
\begin{align*}
\frac{\mathrm{d} }{\mathrm{d} t} & = \omega \frac{\mathrm{d} }{\mathrm{d} \eta} \\
\frac{\mathrm{d}^2}{\mathrm{d} t^2} & = \omega^2 \frac{\mathrm{d}^2 }{\mathrm{d} \eta^2}
\end{align*}
so that system \eqref{eq:linearized-swing-dynamics-perturbed-one-edge} becomes
\begin{align}
\ddot{\phi}(t) & + \left( L + \epsilon E_{ik} \cos(\omega t) \right) \phi(t) = 0 \nonumber\\
 & \updownarrow \nonumber\\
\phi''(\eta) & + \frac{1}{\omega^2} \left( L + \epsilon E_{ik} \cos(t) \right) \phi(\eta) = 0
\end{align}
Defining \(\kappa := \frac{1}{\omega^2}\), we get
\begin{align}
\phi''(\eta) + \left( \kappa L + \kappa \epsilon E_{ik} \cos(t) \right) \phi(\eta) = 0 
\end{align}
For certain values of \(\kappa\), this system is unstable as \(\epsilon \to 0\).  From \eqref{eq:basic-instability-criterion}, this is when
\begin{align}
\pm j \sqrt{\kappa} {\scriptstyle \Omega}_l = \pm j \sqrt{\kappa} {\scriptstyle \Omega}_k \pm j \nonumber\\
\implies \sqrt{\kappa} \abs{{\scriptstyle \Omega}_l +  {\scriptstyle \Omega}_k} = 1  \label{eq:basic-instability-criterion-rescaled-filtered-plus}\\
\implies \sqrt{\kappa} \abs{{\scriptstyle \Omega}_l -  {\scriptstyle \Omega}_k} = 1  \label{eq:basic-instability-criterion-rescaled-filtered-minus}
\end{align}
for some \(l, k  = 0, 1, \dots, n\).  Before proceeding further, we switch our notation back to \(t\) from \(\eta\).  Our rescaled system is thus
\begin{align}
\label{eq:linearized-swing-dynamics-perturbed-one-edge-rescaled}
\ddot{\phi}(t) + \kappa \left( L + \epsilon E_{ik} \cos(t) \right) \phi(t) = 0
\end{align}
This system exhibits the instabilities studied in the previous section for specific critical values of \(\kappa\) (since \(\kappa = \frac{1}{\omega^2}\)).  We investigate the stability of the system (for small \(\epsilon\)) in the vicinity of these \(\kappa\).  Specifically, we are interested in the behavior of the system \eqref{eq:linearized-swing-dynamics-perturbed-one-edge-rescaled} along curves in \(\kappa, \epsilon\) space: 
\begin{equation}
\label{eq:test-curve-in-ke-space}
\kappa(\epsilon) = \kappa_0 + \kappa_1 \epsilon + \kappa_2 \epsilon^2 + \mathcal{O}(\epsilon^3) 
\end{equation}
This formulation allows us to estimate not only the critical values of \(\kappa_0\) for instability, but also the slopes (or widths) of the corresponding Arnold tongues \(\kappa_1\).  Figure \ref{fig:testcurve_diagram} shows visualizations of  \(\kappa(\epsilon)\).

\begin{figure}[htb]
\centering
\includegraphics[width=0.8\columnwidth]{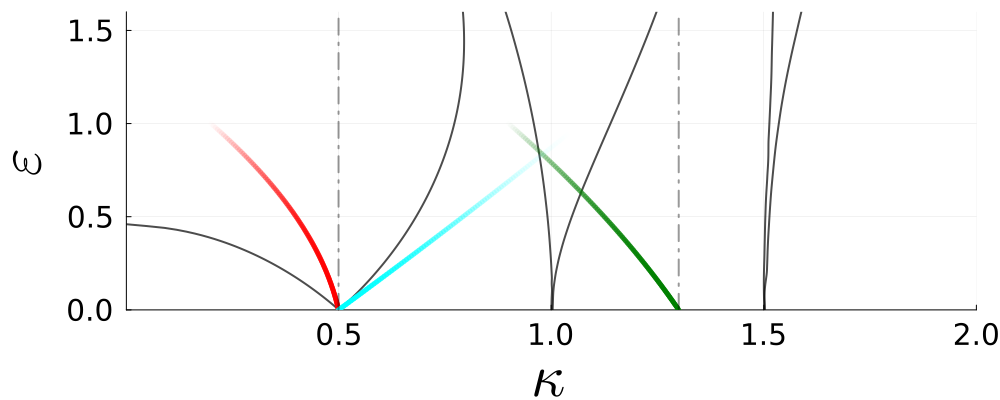}
\caption{\label{fig:testcurve_diagram}Examples of test curves \(\kappa(\epsilon)\) used for multiple time-scale perturbation analysis of system \eqref{eq:linearized-swing-dynamics-perturbed-one-edge-rescaled}, overlaid on top of a stability diagram for a linear time-periodic oscillator in \((\kappa, \epsilon)\) space.  In gray are the Arnold tongues for the system.  The test curves are represented in the diagram as the colored streaks, and algebraically using Equation \eqref{eq:test-curve-in-ke-space}.  We expect the stability of the system at points on the test curves for small \(\epsilon\) to depend on both the \(\kappa\)-intercept \(\kappa_0\) and the slope \(\kappa_1\) measured from the vertical.  Along the green curve (with \(\kappa_0 = 1.3\)), the system is stable for all possible slopes \(\kappa_1\).  For the curves with \(\kappa_0 = 0.5\), the stability of the system depends on the slope \(\kappa_1\) -- the system is unstable at all points on the red curve and stable at all points on the cyan curve.  The objective of the multiple time-scale analysis is to estimate both the critical points \(\kappa_0\) and the critical slopes \(\kappa_1\), or their equivalents in the original scaling in the original \((\omega, \epsilon)\) space (Equation \eqref{eq:linearized-swing-dynamics-perturbed-one-edge}).}
\end{figure}

We treat the configuration of the system \(\phi(t)\) as a function of two time scales, a fast time \(t\) and a slow time scale \(\tau := \epsilon t\).  This allows us to study the effect of the slow time behavior of the system on the fast time dynamics.  Specifically, we set
\begin{equation}
\label{eq:two-time-scale-formulation}
\Phi(t,\tau; \epsilon) := \phi(t;\epsilon)
\end{equation}

Note that it is always possible to find a \(\Phi(t, \tau; \epsilon)\) that satisfies this relation -- the trivial case is when \(\Phi\) does not depend on its second argument \(\tau\).  Since this decomposition is non-unique, \(\Phi(t, \tau; \epsilon)\) is not well-defined yet, and will be uniquely determined in the following analysis.

We can now expand \(\Phi(t, \tau; \epsilon)\) in powers of \(\epsilon\):
\begin{align}
\label{eq:swing-dynamics-2nd-order-regular-perturbation-multiscale}
\Phi(t, \tau; \epsilon) = \sum\limits_{n=0}^N \epsilon^{n} \Phi_n(t, \tau) + \mathcal{O}(\epsilon^{N+1})
\end{align}

We relate the derivatives as follows:
\begin{align*}
\frac{\mathrm{d} \phi}{\mathrm{d} t}(t; \epsilon) & = \frac{\mathrm{d} }{\mathrm{d} t} \Phi(t, \tau; \epsilon) = \frac{\partial \Phi}{\partial t} + \frac{\partial \Phi}{\partial \tau} \frac{\mathrm{d} \tau}{\mathrm{d} t} \\
&  = \frac{\partial \Phi}{\partial t} + \epsilon \frac{\partial \Phi}{\partial \tau} \\
\frac{\mathrm{d}^2 \phi}{\mathrm{d}t^2}(t; \epsilon) & = \frac{\partial^2 \Phi}{\partial t^2} + 2 \epsilon \frac{\partial \Phi}{\partial t} \frac{\partial \Phi}{\partial \tau} + \frac{\partial^2 \Phi}{\partial \tau^2}
\end{align*}

or, in more compact notation
\begin{align*}
\partial_1 (\cdot) & := \frac{\partial (\cdot)}{\partial t},\ \partial_2 (\cdot) := \frac{\partial (\cdot)}{\partial \tau} \\
\implies \frac{\mathrm{d} }{\mathrm{d} t} \phi  & = \partial_1 \Phi + \epsilon \partial_2 \Phi \\
\frac{\mathrm{d}^2 \phi}{\mathrm{d}t^2} & = \partial_1^2 \Phi + 2 \epsilon \partial_1 \partial_2 \Phi + \partial_2^2 \Phi
\end{align*}

Together, \eqref{eq:linearized-swing-dynamics-perturbed-one-edge-rescaled} is transformed to
\begin{align}
\label{eq:linearized-swing-dynamics-perturbed-one-edge-transformed}
  \partial_1^2 \Phi + (\kappa_0 + \kappa_1 \epsilon + \dots) & (L + E_{ik} \cos(t)) \Phi = \nonumber\\
  & - 2 \epsilon \partial_1 \partial_2 \Phi - \epsilon^2 \partial_2^2 \Phi
\end{align}
Substituting for \(\kappa\) from equation \eqref{eq:test-curve-in-ke-space}, for \(\Phi(t, \tau; \epsilon)\) from the series expansion \eqref{eq:swing-dynamics-2nd-order-regular-perturbation-multiscale} and collecting powers of \(\epsilon\), we get to \(\mathcal{O}(\epsilon)\):
\begin{small}
\begin{align}
\mathcal{O}(\epsilon^0): \quad  & \partial_1^2 \Phi_0(t,\tau) + \kappa_0 L \Phi_0(t,\tau) = 0 \label{eq:multiscale-system-Oeps0}\\
\mathcal{O}(\epsilon^1): \quad  & \partial_1^2 \Phi_1(t,\tau) + \kappa_0 L \Phi_1(t,\tau) = \nonumber\\
& \quad \text{-} \left( \kappa_1 L + \kappa_0 E_{ik} \cos(t) + 2 \partial_1 \partial_2 \right) \Phi_0(t,\tau) \label{eq:multiscale-system-Oeps1}
\end{align}
\end{small}

As in the regular perturbation analysis \eqref{eq:swing-dynamics-2nd-order-regular-perturbation-Oeps0-sol}, the \(\mathcal{O}(\epsilon^{0})\) system is a linear, unforced harmonic oscillator whose solution is a matrix exponential:
\begin{equation}
\label{eq:swing-dynamics-1st-order-multiple-scale-sol-Oeps1}
\Phi_0(t, \tau) = V \left( e^{j \sqrt{\kappa_0} \Omega t} \mathbf{c}(\tau) + e^{\text{-}j \sqrt{\kappa_0 } \Omega t} \mathbf{d}(\tau) \right)
\end{equation}
However, the ``initial conditions'' \(\mathbf{c}(\tau)\) and \(\mathbf{d}(\tau)\) are functions of the slow time variable \(\tau\) and yet to be determined.

To analyze the \(\mathcal{O}(\epsilon^{1})\) system, it will be useful to work in the basis of the eigenvectors of \(L\).  With \(L V = V \Omega^2\) and \(V^{\star} = V^{\text{-}1}\),
\begin{align}
  \label{eq:multiple-scale-sol-Oeps1-spectral}
  V^{\star} \Phi_0(t,\tau) & = e^{j \sqrt{\kappa_0} \Omega t} \mathbf{c}(\tau) + e^{\text{-}j \sqrt{k_0}\Omega t} \mathbf{d}(\tau) 
\end{align}
and the \(\mathcal{O}(\epsilon^1)\) system is
\begin{align}
\partial_1^2 (V^{\star} \Phi_1(t,\tau)) + & \kappa_0 V^{\star} L V (V^{\star} \Phi_1(t,\tau)) = \nonumber\\
& \text{-} V^{\star} (\kappa_1 L + \kappa_0  E_{ik} \cos(t) + 2 \partial_1 \partial_2) V (V^{\star} \Phi_0).
\label{eq:swing-dynamics-2nd-order-multiple-scale-Oeps1-eigenbasis}
\end{align}
Substituting for \(V^{\star} \Phi_0\) from \eqref{eq:multiple-scale-sol-Oeps1-spectral}, the additive forcing term becomes
\begin{align}
& - V^{\star} \left( \kappa_1 L + \kappa_0 E_{ik} \cos(t) + 2 \partial_1 \partial_2 \right) V \left( V^{\star} \Phi_0  \right) \nonumber\\
= & - \left( \kappa_1 \Omega^2 + \kappa_0 V^{\star} E_{ik} V \cos(t) + 2 \partial_1 \partial_2 I \right) \times \nonumber\\
& \quad (e^{j \sqrt{\kappa_0} \Omega t} \mathbf{c}(\tau) + e^{\text{-}j \sqrt{\kappa_0}\Omega t} \mathbf{d}(\tau)) \nonumber\\
= & - e^{j \sqrt{\kappa_0}\Omega t}\left( \kappa_1 \Omega^2\, \mathbf{c}(\tau) + 2 j \sqrt{\kappa_0} \Omega\, \mathbf{c}'(\tau) \right) \nonumber\\
& - e^{\text{-}j \sqrt{\kappa_0 } \Omega t} \left( \kappa_1 \Omega^2\, \mathbf{d}(\tau) - 2 j \sqrt{\kappa_0 } \Omega\, \mathbf{d}'(\tau) \right) \nonumber\\
& - \frac{\kappa_0}{2} V^{\star} E_{ik} V \left( e^{j (\sqrt{\kappa_0 }\Omega + I)t} + e^{(\sqrt{\kappa_0 }\Omega - I) t} \right) \mathbf{c}(\tau) \nonumber\\
& + \frac{\kappa_0}{2} V^{\star} E_{ik} V \left( e^{\text{-}j (\sqrt{\kappa_0 }\Omega + I)t} + e^{\text{-}j (\sqrt{\kappa_0 }\Omega - I) t} \right) \mathbf{d}(\tau)
\end{align}
where we have used the facts that
\begin{itemize}
\item \(e^{j \sqrt{\kappa_0} \Omega t}\) commutes with \(\kappa_1 \Omega^2\) by the properties of the matrix exponential, and
\item \(\cos(t) = \frac{1}{2} \left( e^{j t} + e^{\text{-} j t} \right)\).
\end{itemize}

For stability, we require both the fast and slow time dynamics to be stable.  The former requires choices of \(\mathbf{c}(\tau)\) and \(\mathbf{d}(\tau)\) for which the forcing terms that cause harmonic resonance in system \eqref{eq:multiscale-system-Oeps1} (and are \emph{secular}) vanish.  This uniquely determines \(\Phi_0(0, \tau)\).  The latter imposes conditions on \(\mathbf{c}(\tau)\) and \(\mathbf{d}(\tau)\), and in turn on the system parameters \(\omega\) and \(\epsilon\).

We make the following observations:
\begin{enumerate}
\item Every term involving \(e^{j \sqrt{\kappa_0} \Omega t}\) or \(e^{\text{-}j \sqrt{\kappa_0} \Omega t}\) is secular, since these are modes of the unforced system.
\item If \(\kappa_0\) corresponds to \(\omega\) at the base of an Arnold tongue (\emph{i.e.} is ``critical''), then \(j \sqrt{\kappa_0} {\scriptstyle \Omega}_l = \pm j \sqrt{\kappa_0} {\scriptstyle \Omega}_m \pm 1\) for some \(l, m\).  Then the corresponding term involving \(\exp \left( j \left( \pm \sqrt{\kappa_0} {\scriptstyle \Omega}_m \pm 1 \right) t \right)\) is secular.
\end{enumerate}

We first consider the simpler case, \emph{i.e} of non-critical \(\kappa_0\).
\subsubsection{In the stable region}
\label{sec:multiple-scale-analysis-stable-region}

When \(\kappa_0\) is not critical, all forcing terms causing harmonic resonance in system \eqref{eq:multiscale-system-Oeps1} are of the form
\begin{align*}
& {\scriptstyle \Omega}_l \left[ c_l'(\tau) - \frac{j}{2} \frac{\kappa_1}{\kappa_0} \kappa_0^{\frac{1}{2}} {\scriptstyle \Omega}_l c_l(\tau) \right] \exp \left( j \kappa_0^{\frac{1}{2}} {\scriptstyle \Omega}_l t \right) \\
& {\scriptstyle \Omega}_l \left[ d_l'(\tau) + \frac{j}{2} \frac{\kappa_1}{\kappa_0} \kappa_0^{\frac{1}{2}} {\scriptstyle \Omega}_l d_l(\tau) \right] \exp \left( - j \kappa_0^{\frac{1}{2}} {\scriptstyle \Omega}_l t \right)
\end{align*}
where \(c_l\) is the \(l^{\mathrm{th}}\) component of \(\mathbf{c}(\tau)\). There are \(n\) such components of \(\mathbf{c}\) and \(\mathbf{d}\), corresponding to the exponentials of \(\pm j  \kappa_0^{\frac{1}{2}} {\scriptstyle \Omega}_l t\), \(l = 1, 2, \dots, n\).  By assumption, none of the terms involving \(\exp \left( j \left( \pm \sqrt{ \kappa_0 } {\scriptstyle \Omega}_l \pm 1 \right) t \right)\) are secular.

If we pick \(\mathbf{c}(\tau)\) and \(\mathbf{d}(\tau)\) such that these (independently) vanish, then \(\Phi_1(t, \tau)\) is stable in \(t\), and we fully determine \(\Phi_0(t, \tau)\) in the process.  In the parlance of perturbation theory, we eliminate the secular terms in equation \eqref{eq:multiscale-system-Oeps1}.

The above terms vanish when \(\mathbf{c}(\tau)\) and \(\mathbf{d}(\tau)\) satisfy the ODEs
\begin{align}
& \mathbf{c}'(\tau) - j \frac{\kappa_1}{2 \kappa_0} \kappa_0^{\frac{1}{2}} \Omega \mathbf{c}(\tau) = 0 \nonumber\\
& \mathbf{d}'(\tau) + j \frac{\kappa_1}{2 \kappa_0} \kappa_0^{\frac{1}{2}} \Omega \mathbf{d}(\tau) = 0 \nonumber\\
\implies & \mathbf{c}'(\tau) = j \frac{\kappa_1}{2 \kappa_0^{\frac{1}{2}}} \Omega \mathbf{c}(\tau)\\
& \mathbf{d}'(\tau) = - j \frac{\kappa_1}{2 \kappa_0^{\frac{1}{2}}} \Omega \mathbf{d}(\tau)  \label{eq:secular-term-elimination-stable-region-ODE-1}
\end{align}
The solution to this is
\begin{flalign*}
\mathbf{c}(\tau) = \exp{ \left( j \tfrac{\kappa_1}{2 \kappa_0^{\frac{1}{2}}} \Omega \tau \right) } \mathbf{c}(0), & \quad \mathbf{d}(\tau) = \exp{ \left( \text{-} j \tfrac{\kappa_1}{2 \kappa_0^{\frac{1}{2}}} \Omega \tau \right) } \mathbf{d}(0)
\end{flalign*}
This choice of \(\mathbf{c}(\tau)\) and \(\mathbf{d}(\tau)\) makes \(\Phi_1(t,\tau)\) stable.  Since \(\Omega\) is diagonal with non-negative entries, \(\mathbf{c}(\tau)\) and \(\mathbf{d}(\tau)\) are bounded, and \(\Phi_0(t, \tau)\) is stable as well.  This is the case for all \(\kappa_1\), \emph{i.e.} irrespective of our choice of test curve in \((\kappa, \epsilon)\) space.  This is in line with our expectations from the analysis of section \ref{sec:regular-perturbation-analysis}, since every critical \(\kappa_0\) should satisfy one of the conditions \eqref{eq:basic-instability-criterion-rescaled-filtered-plus} or \eqref{eq:basic-instability-criterion-rescaled-filtered-minus}.  Our test curve \(\kappa(\epsilon)\) corresponds to the green curve in figure \ref{fig:testcurve_diagram}.  When starting from a non-critical \(\kappa_0\), the system remains stable along all curves in the \((\epsilon, \kappa)\) parameter space (and thus in the \((\epsilon, \omega)\) domain) as \(\epsilon \to 0\).

Further, \(\Phi_0(t, \tau)\) is given by
\begin{align}
V^{\star} \left( e^{j \sqrt{\kappa_0 }\Omega \left( 1 + \epsilon \frac{\kappa_1}{2 \kappa_0} \right) t } \mathbf{c}(0) + e^{-j \sqrt{\kappa_0 }\Omega \left( 1 + \epsilon \frac{\kappa_1 }{2 \kappa_0 } \right) t} \mathbf{d}(0) \right)
\label{eq:secular-term-elimination-stable-region-ODE-2}
\end{align}
which is a sinusoidal response at a perturbed frequency.  The change in the frequency is proportional to \(\epsilon\) and the slope \(\kappa_1\).
\subsubsection{At critical frequencies \(\kappa_0\), \(\sqrt{\kappa_0} \abs{{\scriptstyle \Omega}_l \pm {\scriptstyle \Omega}_m} = 1\)}
\label{sec:multiple-scale-analysis-unstable-region-choice-1}

When \(\kappa_0\) is critical, \emph{i.e.} when either \eqref{eq:basic-instability-criterion-rescaled-filtered-plus} or \eqref{eq:basic-instability-criterion-rescaled-filtered-minus} is true, forcing terms in the \(\mathcal{O}(\epsilon)\) dynamics \eqref{eq:multiscale-system-Oeps1} involving \(\exp \left( j \left( \sqrt{\kappa_0} {\scriptstyle \Omega}_l \pm 1 \right) t \right)\) also cause harmonic resonance and are secular.  Suppose one of the relations \(\sqrt{\kappa_0} \abs{ {\scriptstyle \Omega}_l \pm {\scriptstyle \Omega}_m } = 1\) holds for a given pair of indices \(l\) and \(m\).  Restricting our attention to equations in  \eqref{eq:swing-dynamics-2nd-order-multiple-scale-Oeps1-eigenbasis} with forcing terms involving the frequencies \(\sqrt{\kappa_0} {\scriptstyle \Omega}_l\) and \(\sqrt{\kappa_0} {\scriptstyle \Omega}_m\), we get
\begin{align}
& \partial_1^2 \left( V^{\star} \Phi_0(t,\tau)  \right)_m + \kappa_0 {\scriptstyle \Omega}_m^2 \left( V^{\star} \Phi_0(t,\tau) \right)_m = \nonumber\\
& - e^{j \sqrt{\kappa_0}\Omega_m t}\left[ \kappa_1 {\scriptstyle \Omega}^2_m + 2 j \sqrt{\kappa_0} {\scriptstyle \Omega}_m c_m'(\tau) \right] \nonumber\\
& - \tfrac{1}{2} \kappa_0 \mathbf{v}^{\star}_m E_{ik} \mathbf{v}_l c_l(\tau) \nonumber\\
& \times \left( e^{j (\sqrt{\kappa_0 }\Omega_l + 1)t} + e^{j (\sqrt{\kappa_0 }\Omega_l \text{-} 1)t} \right) \label{eq:swing-dynamics-2nd-order-multiple-scale-Oeps1-two-modes-1}\\
& \partial_1^2 \left( V^{\star} \Phi_0(t,\tau)  \right)_l + \kappa_0 {\scriptstyle \Omega}_l^2 \left( V^{\star} \Phi_0(t,\tau) \right)_l = \nonumber\\
& - e^{j \sqrt{\kappa_0}\Omega_l t}\left[ \kappa_1 {\scriptstyle \Omega}^2_l + 2 j \sqrt{\kappa_0} {\scriptstyle \Omega}_l c_l'(\tau) \right] \nonumber\\
& - \tfrac{1}{2} \kappa_0 \mathbf{v}^{\star}_l E_{ik} \mathbf{v}_m c_m(\tau) \nonumber\\
& \times \left( e^{j (\sqrt{\kappa_0 }\Omega_m + 1)t} + e^{j (\sqrt{\kappa_0 }\Omega_m \text{-} 1)t} \right) \label{eq:swing-dynamics-2nd-order-multiple-scale-Oeps1-two-modes-2}
\end{align}

with similar terms involving \(\mathbf{d}(\tau)\) and the frequencies \(\text{-} \sqrt{\kappa_0} {\scriptstyle \Omega}_l\) and \(\text{-} \sqrt{\kappa_0} {\scriptstyle \Omega}_m\).  Here we have assumed that the eigenspaces corresponding to eigenvalues \({\scriptstyle \Omega}_l\) and \({\scriptstyle \Omega}_m\) are \(1\)-dimensional.  The case of eigenvalue multiplicity is handled in \cref{sec:multi-scale-method-with-multiplicity}.

To find conditions on \(c_l(\tau)\) and \(c_m(\tau)\) such that \(\Phi_1(t, \tau)\) is stable in \(t\), we need to collect terms involving like exponentials.  This organization differs between the two cases  \(\sqrt{\kappa_0} \abs{{\scriptstyle \Omega}_l \pm {\scriptstyle \Omega}_m} = 1\).
\paragraph{Case I: \(\sqrt{\kappa_0} \abs{{\scriptstyle \Omega}_l -  {\scriptstyle \Omega}_m} = 1\)}

This condition applies when the difference between the square roots of eigenvalues of the Laplacian \(L\) equals the original parametric forcing frequency \(\omega\), since \(\abs{{\scriptstyle \Omega}_l - {\scriptstyle \Omega}_m} = 1 / \sqrt{\kappa_0} = \omega\).

Without loss of generality, we can pick \({\scriptstyle \Omega}_l \ge {\scriptstyle \Omega}_m\), so that
\begin{itemize}
\item \(\exp \left( j \sqrt{\kappa_0} {\scriptstyle \Omega}_l t \right) = \exp \left( j \left( \sqrt{\kappa_0} {\scriptstyle \Omega}_m + 1 \right) t \right)\) and
\item \(\exp \left( j \sqrt{\kappa_0} {\scriptstyle \Omega}_m t \right) = \exp \left( j \left( \sqrt{\kappa_0} {\scriptstyle \Omega}_l - 1 \right) t \right)\),
\end{itemize}
corresponding to forcing terms in \eqref{eq:swing-dynamics-2nd-order-multiple-scale-Oeps1-two-modes-1} and \eqref{eq:swing-dynamics-2nd-order-multiple-scale-Oeps1-two-modes-2}:
\begin{small}
\begin{align}
& \left[ c_l'(\tau) - j \tfrac{\kappa_1}{2 \kappa_0} \kappa_0^{\frac{1}{2}} {\scriptstyle \Omega}_l c_l(\tau) - j\,\tfrac{1}{4} \kappa_0^{\text{-}\frac{1}{2}} {\scriptstyle \Omega}_l^{\text{-}1} \mathbf{v}_l^{\star} E_{ik} \mathbf{v}_m c_m(\tau) \right] \nonumber\\
& \times \text{-} 2 j {\scriptstyle \Omega}_l e^{ \left( j \sqrt{\kappa_0} {\scriptstyle \Omega}_l t \right)} \label{e_l c_l q:multiscale-analysis-case-1-term-1}\\
& \left[ c_m'(\tau) - j \tfrac{\kappa_1}{2 \kappa_0}  \kappa_0^{\frac{1}{2}} {\scriptstyle \Omega}_m c_m(\tau) - j\tfrac{1}{4}  \kappa_0^{\text{-}\frac{1}{2}} {\scriptstyle \Omega}_m^{\text{-}1} \mathbf{v}_m^{\star} E_{ik} \mathbf{v}_l c_l(\tau) \right] \nonumber\\
& \times \text{-} 2 j {\scriptstyle \Omega}_m e^{\left( j \sqrt{\kappa_0} {\scriptstyle \Omega}_m t \right)} \label{eq:multiscale-analysis-case-1-term-2}
\end{align}
\end{small}
The forcing terms corresponding to all other modes (\({\scriptstyle \Omega} \ne {\scriptstyle \Omega}_l, {\scriptstyle \Omega}_m\)) are of the form \(\left[ c'(\tau) - j \tfrac{\kappa_1}{2 \kappa_0}  \kappa_0^{\frac{1}{2}} {\scriptstyle \Omega} \mathbf{c}(\tau) \right] e^{\left( j \sqrt{\kappa_0} {\scriptstyle \Omega} t \right)}\), for which the analysis from Section \ref{sec:multiple-scale-analysis-stable-region} carries over.

We pick \(c_l(\tau)\) and \(c_m(\tau)\) such that both these terms vanish.
\begin{small}
\begin{align}
& c_l'(\tau) \,\text{-}\, j \tfrac{\kappa_1}{2 \kappa_0}  \kappa_0^{\frac{1}{2}} {\scriptstyle \Omega}_l c_l(\tau) \,\text{-}\, j \tfrac{1}{4}  \kappa_0^{\text{-}\frac{1}{2}} {\scriptstyle \Omega}_m^{\text{-}1} \mathbf{v}_l^{\star} E_{ik} \mathbf{v}_m c_m(\tau) = 0 \label{eq:multiscale-analysis-case-1-eq-1}\\
& c_m'(\tau) \,\text{-}\, j \tfrac{\kappa_1}{2 \kappa_0}  \kappa_0^{\frac{1}{2}} {\scriptstyle \Omega}_m c_m(\tau) \,\text{-}\, j \tfrac{1}{4}  \kappa_0^{\text{-}\frac{1}{2}} {\scriptstyle \Omega}_l^{\text{-}1} \mathbf{v}_m^{\star} E_{ik} \mathbf{v}_l c_l(\tau) = 0 \label{eq:multiscale-analysis-case-1-eq-2}
\end{align}
\end{small}
The complex conjugate terms combine similarly:
\begin{itemize}
\item \(\exp \left( - j \sqrt{\kappa_0} {\scriptstyle \Omega}_l t\right) = \exp \left(  j \left( - \sqrt{\kappa_0} {\scriptstyle \Omega}_m - 1 \right) t \right)\) and
\item \(\exp \left( - j \sqrt{\kappa_0} {\scriptstyle \Omega}_m t \right) =  \exp \left( j \left( - \sqrt{\kappa_0} {\scriptstyle \Omega}_l + 1 \right) t \right)\),
\end{itemize}
giving us complex conjugates of equations \eqref{eq:multiscale-analysis-case-1-eq-1} and \eqref{eq:multiscale-analysis-case-1-eq-2} but involving \(d_l(\tau)\) and \(d_m(\tau)\).

We can rewrite this system of equations as:
\begin{small}
\begin{align}
\begin{bmatrix} c_{l}'(\tau) \\ c_{m}'(\tau) \end{bmatrix} & = \frac{j}{2 \sqrt{\kappa_0}} \underbrace{\begin{bmatrix}
\kappa_1 {\scriptstyle \Omega}_l & \frac{\kappa_0 (\mathbf{v}_l^{\star} E_{ik} \mathbf{v}_m)}{2 {\scriptstyle \Omega}_l}  \\
\frac{\kappa_0 (\mathbf{v}_m^{\star} E_{ik} \mathbf{v}_l)}{2 {\scriptstyle \Omega}_m}  & \kappa_1 {\scriptstyle \Omega}_m \end{bmatrix}}_A 
\begin{bmatrix} c_{l}(\tau) \\ c_{m}(\tau) \end{bmatrix} \label{eq:multiscale-analysis-case-1-ode-system}\\
\text{or } \mathbf{c}'(\tau) & = \frac{j}{2 \sqrt{\kappa_0}} A\  \mathbf{c}(\tau) \nonumber
\end{align}
\end{small}
This choice of \(\mathbf{c}(\tau)\) makes \(\Phi_1(t, \tau)\) stable in \(t\). For bounded slow-time behavior of \(\Phi_0(t, \tau)\), we require \(\mathbf{c}(\tau)\) to be stable, or the eigenvalues of \(A\) to be real (because of the \(j\) prefactor). The characteristic polynomial of \(A\) is
\begin{align*}
& y^2 - \underbrace{\left( {\scriptstyle \Omega}_l + {\scriptstyle \Omega}_m \right) \kappa_1}_{\text{trace}(A)} y \\
& + \underbrace{\left( \kappa_1^2 {\scriptstyle \Omega}_l {\scriptstyle \Omega}_m - (\mathbf{v}_m^{\star} E_{ik} \mathbf{v}_l)^2 \frac{\kappa_0^2}{4 {\scriptstyle \Omega}_m {\scriptstyle \Omega}_l} \right)}_{\text{determinant}(A)} = 0.
\end{align*} 
Requiring its discriminant to be real, we get
\begin{align}
& \left( {\scriptstyle \Omega}_l + {\scriptstyle \Omega}_m \right)^2 \kappa_1^2 - 4 \kappa_1^2 {\scriptstyle \Omega}_1 {\scriptstyle \Omega}_m \\
& + (\mathbf{v}_m^{\star} E_{ik} \mathbf{v}_l)^2 \frac{\kappa_0^2}{{\scriptstyle \Omega}_m {\scriptstyle \Omega}_l} \ge 0 \nonumber\\
\implies & \left( {\scriptstyle \Omega}_l - {\scriptstyle \Omega}_m \right)^2 \kappa_1^2 + (\mathbf{v}_m^{\star} E_{ik} \mathbf{v}_l)^2 \frac{\kappa_0^2}{{\scriptstyle \Omega}_m {\scriptstyle \Omega}_l} \ge 0 \label{eq:multiscale-analysis-case-1-criterion}
\end{align}
This discriminant is always non-negative, implying that there is no instability at \((\kappa_0, \epsilon)\), or along any curve \(\kappa(\epsilon) = \kappa_0 + \kappa_1 \epsilon\).  Despite our initial assessment, the condition \(\sqrt{\kappa_0} \abs{{\scriptstyle \Omega}_l - {\scriptstyle \Omega}_m} = 1\) does not cause instability.  This is one of the sources of false positives in the regular perturbation analysis of section \ref{sec:regular-perturbation-analysis}.  An alternative derivation of this result is provided by the Krein-Gel'fand-Lidskii theorem~\cite{yakubovich1975periodic}, in terms of the stability properties of perturbed time-dependent Hamiltonian systems.

\bigskip
\paragraph{Case II: \(\sqrt{\kappa_0} \abs{{\scriptstyle \Omega}_l + {\scriptstyle \Omega}_m} = 1\)}

This condition applies when the sum of the square roots of eigenvalues of the Laplacian \(L\) is the original parametric forcing frequency \(\omega\), since \(\abs{{\scriptstyle \Omega}_l + {\scriptstyle \Omega}_m} = \frac{1}{\sqrt{\kappa_0}} = \omega\).

By a similar analysis, the relevant secular terms are obtained by combining the coefficients of
\begin{itemize}
\item \(\exp \left( j \sqrt{\kappa_0} {\scriptstyle \Omega}_l t \right) = \exp \left( j \left( - \sqrt{\kappa_0} {\scriptstyle \Omega}_m + 1 \right) t \right)\)
\item \(\exp \left( - j \sqrt{\kappa_0} {\scriptstyle \Omega}_m t \right) = \exp \left( j \left( \sqrt{\kappa_0} {\scriptstyle \Omega}_l - 1 \right) t \right)\),
\end{itemize}
which are
\begin{small}
\begin{align}
& \left[ c_l'(\tau) - \tfrac{j}{2} \tfrac{\kappa_1}{\kappa_0} \sqrt{\kappa_0} {\scriptstyle \Omega}_l c_l(\tau) - \tfrac{j}{4} \tfrac{1}{\sqrt{\kappa_0} {\scriptstyle \Omega}_m} \mathbf{v}_l^{\star} E_{ik} \mathbf{v}_m d_m(\tau) \right] \nonumber\\
& \times \text{-} 2 {\scriptstyle \Omega}_l j e^{\left( j \sqrt{\kappa_0} {\scriptstyle \Omega}_l t \right)} \label{eq:multiscale-analysis-case-2-term-1}\\
& \left[ d_m'(\tau) + \tfrac{j}{2} \tfrac{\kappa_1}{\kappa_0} \sqrt{\kappa_0} {\scriptstyle \Omega}_m d_m(\tau) + \tfrac{j}{4} \tfrac{1}{\sqrt{\kappa_0} {\scriptstyle \Omega}_m} \mathbf{v}_m^{\star} E_{ik} \mathbf{v}_l c_l(\tau) \right] \nonumber\\
& \times \text{-} 2 {\scriptstyle \Omega}_m j e^{\left( j \sqrt{\kappa_0} {\scriptstyle \Omega}_m t \right)} \label{eq:multiscale-analysis-case-2-term-2}
\end{align}
\end{small}
The complex conjugate terms combine similarly:
\begin{itemize}
\item \(\exp \left( - j \sqrt{\kappa_0} {\scriptstyle \Omega}_l t \right) = \exp \left( j \left( \sqrt{\kappa_0} {\scriptstyle \Omega}_m - 1 \right) t \right)\)
\item \(\exp \left( j \sqrt{\kappa_0} {\scriptstyle \Omega}_m t \right) = \exp \left( j \left( - \sqrt{\kappa_0} {\scriptstyle \Omega}_l + 1 \right) t \right)\),
\end{itemize}
giving us complex conjugates of terms \eqref{eq:multiscale-analysis-case-2-term-1} and \eqref{eq:multiscale-analysis-case-2-term-2}, but involving \(c_m(\tau)\) and \(d_l(\tau)\).

We pick \(c_l(\tau)\) and \(d_m(\tau)\) so that these terms vanish, which gives us the system of equations
\begin{align}
\begin{bmatrix} c_{l}'(\tau) \\ d_{m}'(\tau) \end{bmatrix} & = \frac{j}{2 \sqrt{\kappa_0}} \underbrace{\begin{bmatrix}
\kappa_1 {\scriptstyle \Omega}_l & \frac{\kappa_0 (\mathbf{v}_l^{\star} E_{ik} \mathbf{v}_m)}{2 {\scriptstyle \Omega}_l}   \\
- \frac{\kappa_0 (\mathbf{v}_m^{\star} E_{ik} \mathbf{v}_l)}{2 {\scriptstyle \Omega}_m}  & - \kappa_1 {\scriptstyle \Omega}_m \end{bmatrix}}_A 
\begin{bmatrix} c_l(\tau) \\ d_m(\tau) \end{bmatrix} \label{eq:multiscale-analysis-case-2-ode-system}\\
\text{or } \mathbf{c}'(\tau) & = \frac{j}{2 \sqrt{\kappa_0}} A\  \mathbf{c}(\tau). \nonumber
\end{align}
As before, this choice of \(\mathbf{c}(\tau)\) makes \(\Phi_1(t, \tau)\) stable in \(t\).  For bounded slow-time behavior of \(\Phi_0(t, \tau)\), we require \(\mathbf{c}(\tau)\) to be stable, or the eigenvalues of \(A\) to be real (because of the \(j\) prefactor).  For the discriminant of the characteristic polynomial of \(A\) to be non-negative, we require
\begin{align*}
& \left( {\scriptstyle \Omega}_l - {\scriptstyle \Omega}_m \right)^2 \kappa_1^2 + 4 {\scriptstyle \Omega}_l {\scriptstyle \Omega}_m \kappa_1^2 \\
& - \left( \mathbf{v}_m^{\star} E_{ik} \mathbf{v}_l \right)^2 \frac{\kappa_0^2}{{\scriptstyle \Omega}_l {\scriptstyle \Omega}_m} \ge 0 \\
\implies & \left( {\scriptstyle \Omega}_l + {\scriptstyle \Omega}_m \right)^2 \kappa_1^2 - \left( \mathbf{v}_m^{\star} E_{ik} \mathbf{v}_l \right)^2 \frac{\kappa_0^2}{{\scriptstyle \Omega}_l {\scriptstyle \Omega}_m} \ge 0 \\
\implies &  \kappa_1^2 \ge \left( \mathbf{v}_m^{\star} E_{ik} \mathbf{v}_l \right)^2 \frac{\kappa_0^2}{\left( {\scriptstyle \Omega}_l + {\scriptstyle \Omega}_m \right)^2 {\scriptstyle \Omega}_l {\scriptstyle \Omega}_m} 
\end{align*}
Since \({\scriptstyle \Omega}_l + {\scriptstyle \Omega}_m = \frac{1}{\sqrt{\kappa_0}}\), this stability condition can be simplified further, to give
\begin{equation}
\label{eq:multiscale-analysis-case-2-criterion}
\abs{\kappa_1} \ge \abs{\mathbf{v}_m^{\star} E_{ik} \mathbf{v}_l} \frac{\kappa_0^{3/2}}{\left( {\scriptstyle \Omega}_l {\scriptstyle \Omega}_m \right)^{1/2}}
\end{equation}
This suggests that there is a critical slope \(\kappa_1\) below which the solution at \((\kappa_0 + \kappa_1 \epsilon, \epsilon)\) is unstable.  In figure \ref{fig:testcurve_diagram}, this corresponds to one of the two test curves (red and cyan) originating at a point of instability, with the critical slope \(\kappa_1\) determining stability as \(\epsilon \to 0\).
\subsubsection{Interpreting the multiple-scale analysis stability criteria}
Our findings from the first-order multiple scale perturbation analysis are presented in table \ref{tab:multi-scale-findings-k-eps}. 

\begin{table}[htbp]
\centering
\begin{tabular}{lll}
\toprule
Condition on \(\kappa_0\) & Condition on \(\kappa_1\) & Stability\\
\midrule
\(\sqrt{\kappa_0} {\scriptstyle \Omega}_l \pm \sqrt{\kappa_0} {\scriptstyle \Omega}_m \ne \pm 1\) & - & Always stable\\
\(\sqrt{\kappa_0} {\scriptstyle \Omega}_l - \sqrt{\kappa_0} {\scriptstyle \Omega}_m = \pm 1\) & - & Always stable\\
\(\sqrt{\kappa_0} {\scriptstyle \Omega}_l + \sqrt{\kappa_0} {\scriptstyle \Omega}_m = 1\) & \(\abs{\kappa_1} \ge \frac{\kappa_0^{3/2} \abs{\mathbf{v}_m^{\star} E_{ik} \mathbf{v}_l}}{\left( {\scriptstyle \Omega}_l {\scriptstyle \Omega}_m \right)^{1/2}}\) & Stable\\
 & \(\abs{\kappa_1} < \frac{\kappa_0^{3/2} \abs{\mathbf{v}_m^{\star} E_{ik} \mathbf{v}_l}}{\left( \Omega_l \Omega_m \right)^{1/2}}\) & Unstable\\
\bottomrule
\end{tabular}
\caption{\label{tab:multi-scale-findings-k-eps}Conditions for stability of the system \eqref{eq:linearized-swing-dynamics-perturbed-one-edge-rescaled}, which represents linearized swing dynamics on a graph where one edge weight is varied periodically.  Here, \(\kappa\) lies on a test curve \(\kappa = \kappa_0 + \epsilon \kappa_1 + \mathcal{O}(\epsilon^2)\), and we impose conditions (derived in \cref{sec:vector-valued-multiple-scale-perturbation-analysis}) on \(\kappa_0\) and \(\kappa_1\) for stability at low forcing amplitudes \(\epsilon\) (\emph{i.e.} as \(\epsilon \to 0\)).  In this limit, the system is unstable at a discrete set of parameters \(\kappa_0\) that depends on the eigenvalues and eigenvectors of the graph Laplacian.}

\end{table}

Since \(\kappa_1\) represents the slope of the curve we are following in the \((\kappa, \epsilon)\) plane at \((0, \kappa_0)\), the critical value of \(\kappa_1\) corresponds to the tongue boundary (when the tongue exists).  We can reinterpret this in the \((\omega, \epsilon)\) plane via the scaling rule \(\kappa = 1 / \omega^2\).  The curve \(\kappa(\epsilon) = \kappa_0 + \epsilon \kappa_1 + \mathcal{O}(\epsilon^2)\) transforms to \(\omega(\epsilon) = \omega_0 + a \epsilon + \mathcal{O}(\epsilon^2)\), where \(\kappa_0 = \frac{1}{\omega_0^2}\) and \( a \) is \(\mathrm{d} \omega(\epsilon) / \mathrm{d} \epsilon \vert_{\epsilon = 0}\).  Then
\begin{align}
\kappa_1 & := \left. \frac{\mathrm{d} \kappa(\epsilon)}{\mathrm{d} \epsilon} \middle|_{\epsilon=0} \right. \nonumber\\
& = \left. \frac{\mathrm{d} }{\mathrm{d} \epsilon} \frac{1}{\omega(\epsilon)^2} \middle|_{\epsilon=0} \right. \nonumber\\
& = - \left[ \frac{2}{\omega(\epsilon)^3 } \frac{\mathrm{d} \omega(\epsilon)}{\mathrm{d} \epsilon} \right]_{\epsilon = 0} = - \frac{2}{\omega_0^3} \left. \frac{\mathrm{d} \omega(\epsilon)}{\mathrm{d} \epsilon} \middle|_{\epsilon = 0} \right. \nonumber\\
& = - \frac{2}{\omega_0^3} a \label{eq:k-to-omega-transformation-1}
\end{align}
The stability condition \eqref{eq:multiscale-analysis-case-2-criterion} on \(\kappa_1\) can then be expressed in terms of \(\omega\), since \(\kappa_0 = 1 / \omega_0^2\):
\begin{align}
\abs{\kappa_1} & \ge \abs{\mathbf{v}_m^{\star} E_{ik} \mathbf{v}_l} \frac{\kappa_0^{3/2}}{\left( {\scriptstyle \Omega}_l {\scriptstyle \Omega}_m \right)^{1/2}} \nonumber\\
\implies \abs{- \frac{2}{\omega_0^3} a} & \ge \abs{\mathbf{v}_m^{\star} E_{ik} \mathbf{v}_l} \frac{1}{\omega_0^3 \left( {\scriptstyle \Omega}_l {\scriptstyle \Omega}_m \right)^{1/2}} \nonumber\\
\implies \abs{a} & \ge \frac{\abs{\mathbf{v}_m^{\star} E_{ik} \mathbf{v}_l}}{2 \left( {\scriptstyle \Omega}_l {\scriptstyle \Omega}_m \right)^{1/2}} \label{eq:k-to-omega-transformation-2}
\end{align}
From \eqref{eq:k-to-omega-transformation-1} and \eqref{eq:k-to-omega-transformation-2}, we can state the \emph{instability} criteria for solutions of the system \eqref{eq:linearized-swing-dynamics-perturbed-one-edge} for \(\epsilon \to 0\) at a point on the curve \(\omega(\epsilon) = \omega_0 + a \epsilon + \mathcal{O}(\epsilon^2)\) as
\begin{align}
& \omega_0 = {\scriptstyle \Omega}_l + {\scriptstyle \Omega}_m,\ l, m \in 1, 2, \dots, n \nonumber\\
& \abs{a} < \frac{\abs{\mathbf{v}_m^{\star} E_{ik} \mathbf{v}_l}}{2 \left( {\scriptstyle \Omega}_l {\scriptstyle \Omega}_m \right)^{1/2}} \label{eq:multiscale-analysis-full-criteria-appendix}
\end{align}
\subsection{The multiple scale method with eigenvalue multiplicity}
\label{sec:multi-scale-method-with-multiplicity}

The multiscale analysis of appendix~\ref{sec:vector-valued-multiple-scale-perturbation-analysis} assumes that all eigenvalues of the graph Laplacian have multiplicity \(1\).  This is not true for graphs with symmetries, such as the ring networks considered in Section~\ref{sec:ring-network-susceptibility-analysis}.  The vector-valued multi-scale analysis is still valid, but now Equations~\eqref{eq:swing-dynamics-2nd-order-multiple-scale-Oeps1-two-modes-1} and \eqref{eq:swing-dynamics-2nd-order-multiple-scale-Oeps1-two-modes-2} apply to any pair of eigenvectors \(\mathbf{v}_m\) and \(\mathbf{v}_l\) of the Laplacian \(L\) corresponding to the eigenvalues \({\scriptstyle \Omega}_m^2\) and \({\scriptstyle \Omega}_l^2\) respectively.  Let \(V_m\) and \(V_l\) be matrices whose columns form an orthonormal basis for the eigenspaces corresponding to \({\scriptstyle \Omega}_m^2\) and \({\scriptstyle \Omega}_l^2\) respectively.  Then we have
\begin{align*}
\mathbf{v}_m = V_m \mathbf{z}_m, \quad \mathbf{v}_l = V_l \mathbf{z}_l
\end{align*}
for appropriately sized vectors \(\mathbf{z}_m\) and \(\mathbf{z}_l\).  Further, \(\left\|\mathbf{z}_m\right\| = 1\) since \(\left\|\mathbf{v}_m\right\| = 1\) and \(V_m\) is unitary (similarly \(\mathbf{z}_l\)).
Consequently the instability criterion~\eqref{eq:multiscale-analysis-full-criteria} is modified to
\begin{align}
& \omega_0 = {\scriptstyle \Omega}_l + {\scriptstyle \Omega}_m,\ l, m \in 1, 2, \dots, n \nonumber\\
& \abs{\omega'(0)} < \max_{\left\|\mathbf{z}_m\right\| \text{=} \left\| \mathbf{z}_l \right\| \text{=} 1}\frac{\abs{\mathbf{z}_m^{\star} \left( V_m^{\star} E_{ik} V_l \right) \mathbf{z}_l}}{2 \left( {\scriptstyle \Omega}_l {\scriptstyle \Omega}_m \right)^{1/2}}. \label{eq:multiscale-analysis-full-criteria-with-multiplicity}
\end{align}
This maximum is attained when \(\mathbf{z}_m\) and \(\mathbf{z}_l\) are the first left and right singular vectors of \(V_m^{\star} E_{ik} V_l\) respectively, and the maximum is the \(2\)-induced norm of this matrix:
\begin{align}
\label{eq:multiscale-analysis-full-criteria-with-multiplicity-simplified}
\abs{\omega'(0)} < \frac{\left\|V_m^{\star} E_{ik} V_l\right\|_2}{2 \left( {\scriptstyle \Omega}_l {\scriptstyle \Omega}_m \right)^{1/2}}
\end{align}
\subsection{Multiple scale perturbation analysis: unstable modes}
\label{sec:multi-scale-method-unstable-modes}

In this section we verify that the modes of the unstable response at a critical frequency \(\omega_{lm} = \abs{{\scriptstyle \Omega}_l + {\scriptstyle \Omega}_m}\) (\cref{eq:multiscale-analysis-full-criteria-appendix}) are the corresponding eigenvectors \(\mathbf{v}_l\) and \(\mathbf{v}_m\) of the graph Laplacian.  From \cref{eq:multiple-scale-sol-Oeps1-spectral}, the \(\mathcal{O}(1)\) term in the expansion of the response is given by
\begin{align*}
\Phi_0(t,\tau) = V \left( e^{j \sqrt{\kappa_0 } \Omega t} \mathbf{c}(\tau) +  e^{\text{-}j \sqrt{\kappa_0 }\Omega t} \mathbf{d}(\tau) \right).
\end{align*}
At a critical frequency \(\omega_{lm}\), \(\mathbf{c}(\tau)\) and \(\mathbf{d}(\tau)\) satisfy (from \cref{eq:secular-term-elimination-stable-region-ODE-1} and \cref{eq:multiscale-analysis-case-2-ode-system})
\begin{align*}
\begin{bmatrix} c_k(\tau) \\ d_k(\tau) \end{bmatrix} & = 
   \exp \left( j \frac{\kappa_1 }{2 \kappa_0^{1/2}}\begin{pmatrix}
1 & 0 \\
0 & \text{-}1
 \end{pmatrix} \tau \right) \begin{bmatrix} c_{k}(0) \\ d_{k}(0) \end{bmatrix} &  k \ne l,m \\
 \begin{bmatrix} c_l(\tau) \\ d_m(\tau) \end{bmatrix} & = \exp \left( j \frac{1}{2 \sqrt{\kappa_0 }} A \tau \right) \begin{bmatrix} c_{l}(0) \\ d_{m}(0) \end{bmatrix} \\
 \begin{bmatrix} c_m(\tau) \\ d_l(\tau) \end{bmatrix} & = \exp \left( \text{-}j \frac{1}{2 \sqrt{\kappa_0 }} A \tau \right) \begin{bmatrix} c_{m}(0) \\ d_{l}(0) \end{bmatrix}
\end{align*}
where \(A\) is as in \cref{eq:multiscale-analysis-case-2-ode-system} and has eigenvalues with non-zero imaginary part.  \(c_l\), \(c_m\), \(d_l\) and \(d_m\) are the only components of \(\mathbf{c}\) and \(\mathbf{d}\) that grow unboundedly.  \(\Phi_0 (t,\tau)\) is
\begin{align*}
\Phi_0 (t,\tau) & = \mathbf{v}_l (e^{j \sqrt{\kappa_0 } \Omega_l t} c_l(\tau) + e^{\text{-}j \sqrt{\kappa_0 } \Omega_l t} d_l(\tau)) \\
& + \mathbf{v}_m (e^{j \sqrt{\kappa_0 } \Omega_m t} c_m(\tau) + e^{\text{-}j \sqrt{\kappa_0 } \Omega_m t} d_m(\tau)) \\
& + \text{stable modes.}
\end{align*}
This choice of \(\mathbf{c}\) and \(\mathbf{d}\) makes the secular terms in the \(\mathcal{O}(\epsilon)\) dynamics (\cref{eq:swing-dynamics-2nd-order-multiple-scale-Oeps1-eigenbasis}) vanish, so there are no other unstable modes to \(\mathcal{O}(\epsilon)\).  The unstable response at \(\omega_{lm}\) as predicted by the multiple scale analysis thus involves only the modes \(\mathbf{v}_l\) and \(\mathbf{v}_m\).
\subsection{Eigenvector test for controllability}
\label{sec:eigenvector-test-for-controllability}

We obtain a controllability criterion for the second order system~\eqref{eq:controllability-second-order-illustration} by applying the eigenvector test~\cite{hespanha2018} to it rewritten in first order form.

\begin{lemma}[Eigenvector test for harmonic oscillator networks]
Let \(L \in \mathbb{R}^{n\times n}\) be real-symmetric with \(L \mathbf{v}_m = {\scriptstyle \Omega}_m^2 \mathbf{v}_m\), \({\scriptstyle \Omega}_m \ne 0\).  Let \(\mathcal{B} \in \mathbb{C}^{n \times k}\).  The harmonic oscillator system
\begin{align}
\ddot{\phi} = -L \phi + \mathcal{B} u(t) \label{eq:pbh-test-second-order-system}
\end{align}
is not controllable from the input \(u(t)\) if \(\mathbf{v}_m^{\star} \mathcal{B} = 0\) for some \(m\).
\label{lemma:eigenvector-test-second-order}
\end{lemma}

\begin{proof}
Rewritten in first order form, \cref{eq:pbh-test-second-order-system} is
\begin{align}
\frac{\mathrm{d} }{\mathrm{d} t} \begin{bmatrix} \phi \\ \dot{\phi} \end{bmatrix} = \underbrace{\begin{bmatrix}
0 & I \\
\text{-} L & 0
 \end{bmatrix}}_{\mathcal{L}} \begin{bmatrix} \phi \\ \dot{\phi} \end{bmatrix} + \begin{bmatrix} 0 \\ \mathcal{B} \end{bmatrix} u(t) \label{eq:pbh-test-first-order-system}
\end{align}
For \({\scriptstyle \Omega}_m \ne 0\), the eigenvalues and eigenvectors of the generator \(\mathcal{L}\) are given by the relation
\begin{align*}
\mathcal{L} \begin{pmatrix} \mathbf{v}_{m} \\ \pm j {\scriptstyle \Omega}_{m} \mathbf{v}_m \end{pmatrix} = \pm j {\scriptstyle \Omega}_m \begin{pmatrix} \mathbf{v}_{m} \\ \pm j {\scriptstyle \Omega}_{m} \mathbf{v}_m \end{pmatrix},
\end{align*}
which can be verified by direct computation.

Applying the eigenvector test to \cref{eq:pbh-test-first-order-system}, a left eigenvector of \(\mathcal{L}\) is in the left nullspace of \(\begin{bmatrix} 0 & \mathcal{B} \end{bmatrix}^{\mathrm{T}}\) iff
\begin{align*}
\begin{bmatrix} \mathbf{v}_m^{\star} & \mp j {\scriptstyle \Omega}_m^{\text{-}\star} \mathbf{v}_m^{\star} \end{bmatrix} \begin{bmatrix} 0 \\ \mathcal{B} \end{bmatrix} = 0 \\
\implies \mathbf{v}_m^{\star} \mathcal{B} = 0
\end{align*}
\end{proof}

Applying \cref{lemma:eigenvector-test-second-order} to system~\eqref{eq:controllability-second-order-illustration}, it is uncontrollable from the input \(u_l(t) = e^{j (\Omega_l + \omega)t} c_l\) iff
\begin{align}
  \begin{bmatrix} \mathbf{v}_{m}^{\star}  & \text{-} j {\scriptstyle \Omega}_m^{\text{-}1} \mathbf{v}_{m}^{\star} \end{bmatrix} \underbrace{\begin{bmatrix} 0 & 0 \\ E_{ik} & 0 \end{bmatrix} \begin{bmatrix} \mathbf{v}_{l} \\ j {\scriptstyle \Omega}_l \mathbf{v}_{l} \end{bmatrix}}_{\mathcal{B}} & = 0 \nonumber\\ \implies \mathbf{v}_m^{\star} E_{ik} \mathbf{v}_l & = 0. 
\end{align}
For \({\scriptstyle \Omega}_m = 0\), \(\mathcal{L}\) has a Jordan block of size \(2 k\), where \(L\) has \(k\) zero eigenvalues, and no conclusion can be drawn for general \(\mathcal{B}\).  But when \(L\) and \(E_{ik}\) are symmetric graph Laplacians, as in \cref{eq:controllability-second-order-illustration}, the mode corresponding to the zero eigenvalue of \(L\) is \(\mathbf{v}_m = \mathbbm{1}\).  This mode is not controllable from \(u_l(t)\) either, as evident from the PBH test~\cite{hespanha2018}:
\begin{align*}
& \mathrm{rank} \left[ \begin{array}{c:c}
  \lambda I - \mathcal{L} & \begin{matrix} 0 \\ E_{ik} \mathbf{v}_l \end{matrix} \end{array}
 \right] \\
= & \mathrm{rank} \left[ \begin{array}{cc:c}
\lambda I & \text{-} I & 0\\
L & \lambda I & E_{ik} \mathbf{v}_l
 \end{array} \right] \stackrel{?}{<} n
\end{align*}
When \(\lambda = 0\), this matrix drops rank, since \(\mathbbm{1}^{\star} E_{ik} = 0\), and
\begin{align*}
\begin{bmatrix} 0 & \mathbbm{1}^{\star} \end{bmatrix}\left[ \begin{array}{cc:c}
0 & \text{-} I & 0\\
L & 0 & E_{ik} \mathbf{v}_l
 \end{array} \right] = 0
\end{align*} 
\end{document}